\documentclass[11pt,a4paper]{amsart}
\usepackage{a4wide}
\usepackage{mathrsfs, amssymb,amsmath,url,amsthm,color,comment}
\usepackage[dvipsnames]{xcolor}

\title[Topology is irrelevant]{Topology is irrelevant\\ (in a dichotomy conjecture for infinite domain constraint satisfaction problems)}

\author[L.~Barto]
{Libor Barto}
	\address{Department of Algebra, MFF UK, Sokolovska 83, 186 00 Praha 8, Czech Republic}
	\email{libor.barto@gmail.com}
	\urladdr{http://www.karlin.mff.cuni.cz/~barto/}

\author[M.~Pinsker]
{Michael Pinsker}
	\address{Institut f\"{u}r Diskrete Mathematik und Geometrie, FG Algebra, TU Wien, Austria, and Department of Algebra, Charles University, Czech Republic}    
	\email{marula@gmx.at}
    \urladdr{http://dmg.tuwien.ac.at/pinsker/}

\thanks{
Libor Barto has received funding from the European Research Council
(ERC) under the European Unions Horizon 2020 research and
innovation programme (grant agreement No 771005),
and from the Czech Science Foundation (grant No 13-01832S). Michael Pinsker has received funding from the  Austrian Science Fund (FWF) through  project No P27600, and from the Czech Science Foundation (grant No 18-20123S)}

\theoremstyle{plain}

    \newtheorem{theorem}{Theorem}[section]
    \newtheorem{lemma}[theorem]{Lemma}
    \newtheorem{proposition}[theorem]{Proposition}
    \newtheorem{conjecture}[theorem]{Conjecture}
    \newtheorem{corollary}[theorem]{Corollary}

    \newtheorem{definition}[theorem]{Definition}
    \newtheorem{example}[theorem]{Example}
 
 \newcommand{\co}[1]{{#1}}
 \newcommand{\lib}[1]{{#1}}
 \newcommand{\ignore}[1]{}
 
 \newcommand\blfootnote[1]{%
  \begingroup
  \renewcommand\thefootnote{}\footnote{#1}%
  \addtocounter{footnote}{-1}%
  \endgroup
}

\DeclareMathOperator{\CSP}{CSP}
\DeclareMathOperator{\Pol}{Pol}

\newcommand{\tuple}[1]{\mathbf{#1}}
\newcommand{\clone}[1]{\mathcal{#1}}
\newcommand{\relstr}[1]{\mathbb{#1}}
\newcommand{\NP}{\mathrm{NP}}
\newcommand{\trivclone}{\clone{P}}
\newcommand{\group}[1]{\mathcal{#1}}
\newcommand{\clgr}[1]{\mathrm{Gr}(#1)}
\newcommand{\forb}[1]{\mathscr{#1}}
\newcommand{\To}{\rightarrow}

\usepackage{todonotes}

\usepackage{tikz}

\usetikzlibrary{arrows}
\usetikzlibrary{petri}
\usetikzlibrary{backgrounds}
\usetikzlibrary{decorations.pathmorphing}
\usetikzlibrary{calc}
\tikzstyle{myedge}=[-,semithick]

\tikzstyle{krouzek}=[rounded corners=4, draw=black, fill=white!20]
\tikzstyle{ikrouzek}=[rounded corners=4, draw=black, fill=white!20, ultra thick]

\newcommand{\MyFig}[1]{\begin{figure}[ht]#1\end{figure}}

\begin{document}

\maketitle

\begin{abstract}
%
%


%
%
The tractability conjecture for finite domain Constraint Satisfaction Problems (CSPs) stated that such CSPs are  solvable in polynomial time whenever there is no natural  reduction, in some precise technical sense, from  the 3-SAT problem; otherwise, they are NP-complete. Its recent resolution draws on an algebraic characterization of the conjectured borderline: the CSP of a finite  structure permits no natural reduction from  3-SAT if and only if the  stabilizer of the polymorphism clone of the core of the structure  satisfies some  nontrivial system of identities, and such satisfaction is always witnessed by several specific nontrivial systems of identities which do not depend on the structure. 

The tractability conjecture has been generalized in the above formulation to a certain class of infinite domain CSPs, namely, CSPs of  reducts of finitely bounded homogeneous structures. It was subsequently shown that the conjectured borderline between hardness and tractability, i.e., a natural reduction from 3-SAT, can be characterized for this class by a combination of algebraic and topological properties.  However, it was not known whether the topological component is essential in this characterization. 

We provide a negative answer to this question by proving that the borderline is characterized by one specific algebraic identity, namely the pseudo-Siggers identity $\alpha s(x,y,x,z,y,z) \approx \beta s(y,x,z,x,z,y)$. This accomplishes one of the steps of a proposed strategy for reducing the infinite domain CSP dichotomy conjecture to the finite case.
Our main theorem is also of independent mathematical interest, characterizing a topological property of any $\omega$-categorical core structure (the existence of a continuous homomorphism of a stabilizer of its polymorphism clone to the projections) in purely algebraic terms (the failure of an identity as above).
\end{abstract}




\section{Introduction and Main Results}

\blfootnote{Part of this work appeared in the Proceedings of the 31th Annual IEEE Symposium on Logic in Computer Science  (LICS'16) -- see~\cite{BartoPinsker}.}

The Constraint Satisfaction Problem (CSP) over a relational structure $\relstr{A}$ in a finite language, denoted by $\CSP(\relstr{A})$, is the problem of deciding whether or not a given \emph{primitive positive (pp-)} sentence in the language of $\relstr{A}$  holds in $\relstr{A}$; a sentence is called primitive positive if it is an existentially  quantified conjunction of atomic formulae.
An alternative, combinatorial definition of CSPs is also popular in the literature: $\CSP(\relstr{A})$ is the problem of deciding whether a given finite  relational structure in the same language as $\relstr{A}$ maps homomorphically into $\relstr{A}$.

As an example, consider the structure
\[
\mathbb{K}_3 = (\{1,2,3\}; R), \mbox{ where } R = \{(a,b) \in \{1,2,3\}^2 : a \neq b\}.
\]
A (yes) instance of $\CSP(\mathbb{K}_3)$ is, e.g.,
\[
\exists x_1,x_2,x_3,x_4 \ R(x_1,x_2) \wedge R(x_2,x_3) \wedge R(x_1,x_3) \wedge R(x_2,x_4) \wedge R(x_3,x_4)
\]
which, in the homomorphism viewpoint, corresponds to the input struture
\[
(\{x_1,x_2,x_3,x_4\}; \{(x_1,x_2), (x_2,x_3), (x_1,x_3), (x_2,x_4), (x_3,x_4)\}).
\]
It is readily seen  that $\CSP(\relstr{K}_3)$ is the graph 3-coloring problem.

For CSPs over certain structures, including all finite ones, a computational complexity classification has been conjectured, separating NP-hard problems from polynomial-time solvable ones. In the following, we shall state and discuss this conjecture, and subsequently present an improvement thereof which follows from our results. 

In order to keep the presentation compact, we postpone most definitions to Section~\ref{sect:notation}, and refer also to the survey~\cite{BKW-CSP-Survey} as well as to the shorter~\cite{Barto-survey} for the basics of the finite domain CSP and to the monograph~\cite{Bodirsky-HDR} and the shorter~\cite{Pin15} for the infinite. As a reference for standard notions from model theory and universal algebra, we point to the textbooks~\cite{Hodges,Bergman}.

All structures in the present article are implicitly assumed to be finite or \co{countably infinite}.

\subsection{CSPs over finite structures}
The CSP over a structure with finite domain is clearly contained in the class $\NP$. Some well-known $\NP$-complete problems can be formulated as CSPs over suitable finite structures, such as the 3-coloring problem mentioned above, or 3-SAT, which is the CSP over
\[
\mathbb{L} = (\{0,1\}; R_{000}, R_{001}, R_{011}, R_{111}), \ \mbox{ where } R_{abc} = \{0,1\}^3 \setminus \{(a,b,c)\},
\]
or its variants, e.g., the CSP over
\[
\mathbb{M} = (\{0,1\}; \{(0,0,1),(0,1,0),(1,0,0)\}),
\]
which is the positive 1-in-3-SAT problem. The class of finite domain CSPs also contains interesting problems solvable in polynomial time, such as 2-SAT, HORN-SAT, or systems of linear equations over finite fields. 

It was  conjectured in~\cite{FederVardi} that CSPs over finite structures enjoy a dichotomy in the sense that every such CSP is either $\NP$-complete, or tractable, i.e., solvable in polynomial time. A large amount of attention was brought to confirming or refuting this conjecture; 
 a precise borderline between $\NP$-complete and tractable CSPs was delineated~\cite{JBK} and thenceforth referred to as the \emph{tractability conjecture} or also the \emph{algebraic dichotomy conjecture}, since most of the equivalent formulations are algebraic.
In a recent turn of events, the tractability conjecture was confirmed independently by Andrei Bulatov~\cite{Bulatov-Dichotomy} and Dmitriy Zhuk~\cite{Zhuk-Dichotomy}.

To state it, we first recall several basic facts.
First, every finite structure is homomorphically equivalent to a finite \emph{core}, that is, a structure whose every endomorphism is an automorphism. Since this core is unique up to isomorphism, we refer to it as \emph{the} core of a structure. Moreover,  the CSPs over any two structures which are homomorphically equivalent are equal, and so passing from a finite structure to its core does not result in any loss of information concerning the CSP. 

Second, certain relational constructions on structures yield  computational reductions between their CSPs. One simple such   construction is that of a pp-definition -- if  $\relstr{A}$ and $\relstr{B}$ have the same domain and each relation in $\relstr{B}$ has a \emph{pp-definition without parameters} in $\relstr{A}$ (i.e., a definition by a primitive positive formula over $\relstr{A}$, without additional constant symbols for the elements of the domain of $\relstr{A}$), then the CSP over $\relstr{B}$ has a logspace reduction to the CSP over $\relstr{A}$. This is the case, e.g., for the structures $\relstr{A} := \relstr{M}$ and $\relstr{B} := \relstr{L}$ above, and hence we obtain a reduction of 3-SAT to  positive 1-in-3-SAT. A more  general construction is that of a  \emph{pp-interpretation}, which is, roughly, a higher-dimensional generalization of a pp-definition, and which still allows for a computational reduction:  $\CSP(\relstr{B})$ has a logspace reduction to $\CSP(\relstr{A})$ whenever $\relstr{B}$ has a pp-interpretation without parameters in $\relstr{A}$. Moreover, when $\relstr{A}$ is a core, then this statement is even true for pp-interpretations \emph{with parameters}, i.e., one in which the formulae involved may mention elements of the domain of $\relstr{A}$.

These facts imply that the CSP over a finite structure is $\NP$-hard whenever its core pp-interprets with parameters some structure whose CSP is NP-hard, such as the structures $\relstr{K}_3$, $\relstr{L}$, or $\relstr{M}$.
In fact, these three structures 
not only pp-interpret each other, they pp-interpret all finite structures.
The finite domain tractability conjecture postulated that pp-interpreting all finite structures with parameters in the core is the only source of hardness for finite domain CSPs.

\begin{conjecture}[\cite{JBK}, proved in \cite{Bulatov-Dichotomy,Zhuk-Dichotomy}]\label{conj:dicho_fin}
Let $\relstr{B}$ be a finite structure and let $\relstr{A}$ be the core of $\relstr{B}$. Then
\begin{itemize}
\item $\relstr{A}$ pp-interprets all finite structures with parameters (and thus $\CSP(\relstr{B})$ is NP-complete), or
\item $\CSP(\relstr{B})$ is solvable in polynomial time.
\end{itemize}
\end{conjecture}

The algebraic approach to finite domain CSPs is based on the fact that the pp-interpretability strength of a finite structure $\relstr{A}$ is determined by its set of compatible finitary  operations, the so-called \emph{polymorphism clone} $\Pol(\relstr{A})$ of $\relstr{A}$. 
The first step towards this realization was achieved  in~\cite{Jeavons}, where it was observed that, by classical universal algebraic results~\cite{Geiger,BoKaKoRo},   $\relstr{A}$ pp-defines $\relstr{B}$ without parameters if and only if $\Pol(\relstr{A}) \subseteq \Pol(\relstr{B})$. Subsequently, another classical result of universal algebra~\cite{Bir-On-the-structure} made it possible to generalize this fact to pp-interpretations~~\cite{JBK}.%
\footnote{%
Our presentation here does not follow the historical formulations  of the theory of CSPs, which was first developed mostly on the algebraic side and only later phrased using model theoretic notions such as pp-interpretability; one of the earliest accounts of the latter is~\cite{Bodirsky-HDR}.
}
Namely,
a finite structure $\relstr{A}$ pp-interprets a finite structure $\relstr{B}$ without parameters if and only if there exists a \emph{clone homomorphism} from $\Pol(\relstr{A})$ to $\Pol(\relstr{B})$, that is, a mapping which preserves arities and \emph{identities} (universally quantified equations). This fact   implies that the complexity of $\CSP(\relstr{A})$ only depends on the identities satisfied by operations in $\Pol(\relstr{A})$.   

Since in Conjecture~\ref{conj:dicho_fin}, parameters are permitted (and necessary) in pp-interpretations, we need a corresponding concept for polymorphism clones in order to obtain an algebraic reformulation: the \emph{stabilizer} of $\Pol(\relstr{A})$ by constants $c_1,\ldots,c_n$ is denoted $\Pol(\relstr{A},c_1,\ldots,c_n)$; its elements are those polymorphisms of $\relstr{A}$ which preserve all unary relations $\{c_i\}$. It then follows that a finite structure $\relstr{A}$ pp-interprets a finite structure $\relstr{B}$ with parameters if and only if there exists a \emph{clone homomorphism} from some stabilizer of  $\Pol(\relstr{A})$ to $\Pol(\relstr{B})$. We thus obtain the following algebraic reformulation of the first item of Conjecture~\ref{conj:dicho_fin}.
The clones $\Pol(\relstr{L})$, $\Pol(\relstr{M})$, as well as $\Pol(\co{\relstr{K}_3},0,1,2)$ are \emph{trivial}, i.e., they contain only projections. Let us denote the clone of projections on a $2$-element set by $\trivclone$. The clone of projections on any other set of at least $2$ elements is isomorphic to $\trivclone$.

\begin{theorem}[\cite{Geiger,BoKaKoRo,Bir-On-the-structure}, cf.~\cite{Bodirsky-HDR}] \label{thm:finite-birk}
The following are equivalent for a finite relational structure $\relstr{A}$ with domain $A = \{c_1, \dots, c_n\}$.
\begin{itemize}
\item $\relstr{A}$ pp-interprets all finite structures with parameters.
\item There exists a clone homomorphism from $\Pol(\relstr{A},c_1,\ldots,c_n)$ to $\trivclone$. 
\end{itemize}
\end{theorem}

For the second, algebraic statement of Theorem~\ref{thm:finite-birk} numerous equivalent algebraic criteria have been obtained within the setting of finite structures~\cite{T77,HM88,JBK,MM08,Sig10,KMM14,Cyclic}, making in particular the failure of the condition more easily verifiable: this failure is then usually witnessed by the satisfaction of particular identities in $\Pol(\relstr{A},c_1,\ldots,c_n)$ which cannot be satisfied in $\trivclone$.

\begin{theorem}[\cite{Sig10},\cite{KMM14},\cite{MM08},\cite{Cyclic}]\label{thm:main_fin}
The following are equivalent for a finite core structure $\relstr{A}$ with domain $A=\{c_1,\ldots,c_n\}$. 
\begin{itemize}
\item[(i)] There exists no clone homomorphism $\Pol(\relstr{A},c_1,\ldots,c_n)\To\trivclone$.
\item[(ii)] $\Pol(\relstr{A})$ contains a 6-ary Siggers operation, i.e., an  operation $s$ such that 
\[
s(x,y,x,z,y,z) \approx s(y,x,z,x,z,y).
\] 
\item[(iii)] $\Pol(\relstr{A})$ contains a 4-ary Siggers operation, i.e., an  operation $s$ such that 
$
s(y,x,z,y) \approx s(x,y,x,z).
$ 
\item[(iv)] $\Pol(\relstr{A})$ contains a $k$-ary weak near unanimity operation, $k \geq 2$, i.e., an operation $w$ such that
$
w(y,x,x \ldots ,x) \approx 
w(x,y,x, \ldots, x) \approx
w(x,x,\ldots, x,y)
$ 
\item[(v)] $\Pol(\relstr{A})$ contains a $k$-ary cyclic operation, $k \geq 2$, i.e., an operation $c$ such that
$
c(x_1,x_2, \ldots,x_k) \approx 
c(x_2,\ldots,x_k,x_1)
$ 
\end{itemize}
\end{theorem}

\subsection{CSPs over infinite structures}

When we allow the domain of $\relstr{A}$ to be infinite, the situation changes drastically: every computational decision problem is polynomial-time equivalent to $\CSP(\relstr{A})$ for some $\relstr{A}$~\cite{BodirskyGrohe}! A reasonable assumption on $\relstr{A}$ which sends the CSP back to the class $\NP$, and which still allows to cover many interesting computational problems which cannot be modeled as the CSP of any finite structure, 
is that $\relstr{A}$ is a reduct of a finitely bounded homogeneous structure.

A \emph{finitely bounded homogeneous} structure is, intuitively, highly symmetric and the class of its finite substructures admits a finitary description. Important examples include $(\mathbb{Q};<)$, whose finite substructures are linear orders, and the random graph, whose finite substructures are graphs. A \emph{reduct} of $\relstr{A}$ is a structure whose relations are first order definable in $\relstr{A}$ (without parameters). For example, the structure $(\mathbb{Q}; R)$, where $R$ is the betweenness relation $R = \{(a,b,c) \in \mathbb{Q}^3: (a<b<c) \vee (c<b<a)$\}, is a reduct of $(\mathbb{Q}; <)$ whose CSP is the (NP-complete) betweenness problem from temporal reasoning. 
Our main result is of a purely mathematical nature and concerns the  larger class of  \emph{$\omega$-categorical} structures, that is, structures whose automorphism group acts with only finitely many orbits on tuples of any fixed finite length; they thus have a high degree of symmetry, but not necessarily a finite description, and their CSP can be undecidable. Typical CSPs that are \emph{not} covered even by this larger framework are ``numerical'' CSPs (e.g., linear programming), see~\cite{Bodirsky-HDR}.

Substantial results for CSPs over reducts of finitely bounded homogeneous structures include the full complexity classification of  CSPs over reducts of $(\mathbb{Q};<)$ in \cite{tcsps-journal} (classifying the complexity of problems previously called temporal constraint satisfaction problems), the reducts of the random graph~\cite{BodPin-Schaefer-both} (generalizing Schaefer's theorem for Boolean CSPs to what could be called the propositional logic for graphs), and the reducts of the binary branching $C$-relation~\cite{Phylo-Complexity} (classifying the complexity of problems known as phylogeny CSPs). The methods here include the algebraic methods from the finite, but in addition tools from model theory and Ramsey theory~\cite{BP-reductsRamsey}. Moreover, topological considerations have played a significant role in the development of the theory~\cite{Topo-Birk}, and indeed seem inevitable in a sense, although paradoxically it was believed or at least hoped that they would  ultimately turn out inutile in a general complexity classification. On the other hand, due to the fact that the investigation of infinite domain CSPs is more recent, and the additional technical complications which are to be expected when passing from the finite to the infinite, 
the purely algebraic theory as known in the finite is still quite undeveloped in the infinite. The present work can be seen as the first purely algebraic general result for such CSPs.

Some of the basic facts about  finite domain CSPs have analogues in the infinite. In particular, the notion of a core can be generalized and, by~\cite{Cores-journal}, every $\omega$-categorical structure (in particular, every reduct of a finitely bounded homogeneous structure) is homomorphically equivalent to an $\omega$-categorical core, which is unique up to isomorphism.
Also, pp-interpretations still lead to computational reductions, and the same holds when parameters are used when pp-interpreting in a core.
A generalization of the finite domain tractability conjecture has been formulated by Manuel Bodirsky and the second author of the present article.

\begin{conjecture}[Bodirsky + Pinsker\; 2011; cf.~\cite{BPP-projective-homomorphisms}]\label{conj:dicho}
Let $\relstr{B}$ be a reduct of a finitely bounded  homogeneous structure and let $\relstr{A}$ be the core of $\relstr{B}$. Then
\begin{itemize}
\item $\relstr{A}$ pp-interprets all finite structures with parameters (and thus $\CSP(\relstr{B})$ is NP-complete), or
\item $\CSP(\relstr{B})$ is solvable in polynomial time.
\end{itemize}
\end{conjecture}


Similarly to the finite case, for an  $\omega$-categorical structure $\relstr{A}$ the complexity of $\CSP(\relstr{A})$ only depends, up to polynomial-time reductions, on the polymorphism clone $\Pol(\relstr{A})$ \cite{BodirskyNesetrilJLC},  and it has been shown that pp-interpretations in $\relstr{A}$ can be expressed in terms of this clone~\cite{Topo-Birk}.
However, to this end one additionally needs to take into consideration the natural topological structure of $\Pol(\relstr{A})$. This complication comes, roughly speaking, from the requirement that the dimension of a pp-interpretation needs to be finite, as otherwise it would not give us a computational reduction.

\begin{theorem}[\cite{Topo-Birk}]\label{thm:int}
The following are equivalent for an $\omega$-categorical structure $\relstr{A}$ with domain $A$.
\begin{itemize}
\item $\relstr{A}$ pp-interprets all finite structures with parameters.
\item There exists a \emph{continuous} clone homomorphism from $\Pol(\relstr{A},$ $c_1,\ldots,c_n)$ to $\trivclone$, for some elements $c_1,\ldots,c_n\in{A}$.
\end{itemize}
\end{theorem}

More generally, the complexity of $\CSP(\relstr{A})$ depends,  for an $\omega$-categorical  structure $\relstr{A}$, only on the structure of $\Pol(\relstr{A})$ as a topological clone~\cite{Topo-Birk}. 
A natural, yet unresolved problem when comparing the finite with the $\omega$-categorical setting then
is whether the topological structure of the polymorphism clone is really essential in the infinite, or whether the abstract algebraic structure, i.e., the identities that hold in $\Pol(\relstr{A})$, is sufficient to determine the complexity of the CSP. This problem motivates, in particular, the related concept of \emph{reconstruction} of the topology of a clone from its algebraic structure introduced in~\cite{Reconstruction}, which has its own purely mathematical interest.

\subsection{The result}
We show that the borderline proposed in Conjecture~\ref{conj:dicho} \emph{is} purely algebraic. In particular, if the conjecture is true, then the complexity of CSPs over structures concerned by the conjecture only depends on the identities which hold in the polymorphism clone of their core, rather than the additional topological structure thereof. Moreover, the borderline is characterized by a single
 simple identity generalizing item (ii) in Theorem~\ref{thm:main_fin}; satisfaction of this identity could, similarly to the finite setting, potentially be exploited for proving tractability of the CSP.  We show the following.

\begin{theorem}\label{thm:main}
The following are equivalent for an $\omega$-categorical core structure $\relstr{A}$ with domain $A$. 
\begin{itemize}
\item[(i)] There exists no continuous clone homomorphism $\Pol(\relstr{A},c_1,\ldots,$ $c_n) \To\trivclone$, for any  $c_1,\ldots,c_n\in{A}$.
\item[(ii)] There exists no clone homomorphism $\Pol(\relstr{A},c_1,\ldots,c_n)\To\trivclone$, for any $c_1,\ldots,c_n\in{A}$.
\item[(iii)] $\Pol(\relstr{A})$ contains a pseudo-Siggers operation, i.e., a 6-ary operation $s$ such that 
\[
\alpha s(x,y,x,z,y,z) \approx \beta s(y,x,z,x,z,y)
\] 
for some unary operations $\alpha, \beta \in \Pol(\relstr{A})$.
\end{itemize}
\end{theorem}

Consequently, the missing piece for proving Conjecture~\ref{conj:dicho} can now be stated in purely algebraic terms. 

\begin{conjecture} \label{conj:newdicho}
Let $\relstr{A}$ be the core of a reduct of a finitely bounded  homogeneous structure.
If $\Pol(\relstr{A})$ contains a pseudo-Siggers operation, then $\CSP(\relstr{A})$ is solvable in polynomial time. 
\end{conjecture}

In a proposed strategy~\cite{Pin15} for solving Conjecture~\ref{conj:dicho}, the first step asked to prove that for an $\omega$-categorical structure $\relstr{A}$, the existence of a clone homomorphism $\Pol(\relstr{A})\To\trivclone$ implies the existence of a continuous such homomorphism (cf.~\cite{BPP-projective-homomorphisms}).
The idea then is, roughly speaking, to use Ramsey theory to ``lift" the algorithm for finite structures whose polymorphism clone satisfies  this identity to show that $\CSP(\relstr{A})$ is tractable.

While we do not answer this question, Theorem~\ref{thm:main} gives an answer for the variant which is actually relevant for the CSP: for an $\omega$-categorical core structure $\relstr A$, the existence of a clone homomorphism $\Pol(\relstr{A})\To\trivclone$ implies the existence of a continuous clone homomorphism from \emph{some stabilizer} of $\Pol(\relstr{A})$ to $\trivclone$. Taking into account the existence of non-continuous clone homomorphisms $\Pol(\relstr{A})\To\trivclone$~\cite{BPP-projective-homomorphisms}, even for $\omega$-categorical $\relstr{A}$, as well as the recent discovery of $\omega$-categorical structures $\relstr{A}, \relstr{A}'$ whose polymorphism clones are isomorphic algebraically, but not topologically~\cite{BEKP}, 
it might very well turn out that the answer to the original question is negative, but, as we could then conclude a posteriori, irrelevant for CSPs.

Let us also remark that Theorem~\ref{thm:main} is, by the fact that every $\omega$-categorical structure has a unique $\omega$-categorical core, a statement about all $\omega$-categorical structures, rather than only the structures concerned by Conjecture~\ref{conj:dicho}. Theorem~\ref{thm:main} is therefore remarkable in that non-trivial statements about the class of all $\omega$-categorical structures, other than the fundamental theorem of Ryll-Nardzewski, Engeler, and Svenonius characterizing them,  are practically non-existent.

\subsection{Outline and proof strategy}

The strategy for proving Theorem~\ref{thm:main} is similar to the finite analogue of Theorem~\ref{thm:main} proved by Siggers in~\cite{Sig10} (see also~\cite{KMM14}).
Siggers's reasoning is based on a ``loop lemma'' from Bulatov's work in~\cite{B05} that refines the CSP dichotomy theorem for finite undirected graphs~\cite{HellNesetril}. 

After providing definitions and notation in Section~\ref{sect:notation}, we start our proof in Section~\ref{sect:loop} with a generalization of the loop lemma, the \emph{pseudoloop lemma}, using some of the ideas from~\cite{B05}. Instead of finite graphs we work with infinite objects which we call graph-group-systems, and which consist of a permutation group acting on the vertex set of an infinite graph while preserving its edge relation. In our case, the action of the group will have finitely many orbits in its componentwise action on finite tuples of fixed length, so that in some sense the graph becomes finite modulo the group action. 
It might be some people's cup of tea to imagine such systems as fuzzy finite graphs, whereas others will be inclined to draw lines between ``potatoes" in order to achieve an appropriate visualization. 

Theorem~\ref{thm:main} is derived from the pseudoloop lemma in Section~\ref{sect:main} basically using a standard  technique from universal algebra, albeit adapted to the $\omega$-categorical setting via a compactness argument. 

In Section~\ref{sect:discus} we first discuss possible generalizations of our pseudoloop lemma inspired by theorems about finite graphs, and provide some evidence for the possibility of such generalizations. We then 
 consider possible extensions of our main theorem.

Section~\ref{sect:homos} contains further discussion on clone morphisms in the light of other recent results, in particular from the wonderland of reflections~\cite{wonderland}. In  particular, one of our examples provides the answer to a question in~\cite{BPP-projective-homomorphisms}.

We conclude our work with suggestions for research in Section~\ref{sect:conclusion}.

\section{Definitions and Notation}\label{sect:notation}

Relational structures are denoted by blackboard bold letters, such as $\relstr{A}$, and their domain by the same letter in the plain font, such as $A$.
By a \emph{graph} we mean a relational structure with a single symmetric binary relation.

\subsection{The range of the infinite CSP conjecture}  
A relational structure $\relstr{B}$ is \emph{homogeneous} if every isomorphism between finite induced substructures extends to an automorphism of the entire structure $\relstr{B}$. In that case, $\relstr{B}$ is uniquely determined, up to isomorphism, by its \emph{age}, i.e., the class of its finite induced substructures up to isomorphism. $\relstr{B}$ is \emph{finitely bounded} if its signature is finite and its age is given by a finite set $\forb{F}$ of forbidden finite substructures, i.e., the age consists precisely of those finite structures in its signature which do not (isomorphically) embed any member of $\forb{F}$. A \emph{reduct} of a structure $\relstr{B}$ is a structure $\relstr{A}$ on the same domain which is first-order definable without parameters in $\relstr{B}$. Reducts $\relstr{A}$ of finitely bounded homogeneous structures are \emph{$\omega$-categorical}, i.e., the up to isomorphism unique countable model of their first-order theory. Equivalently, their automorphism groups are \emph{oligomorphic}: they have finitely many orbits in their action on $n$-tuples over $\relstr{A}$, for every finite $n\geq 1$.

\subsection{ pp-formulas and interpretations} A formula is \emph{primitive positive}, in short \emph{pp}, if it contains only equalities, existential quantifiers, conjunctions, and atomic formulas -- in our case, relational symbols. A pp-formula \emph{with parameters} over a structure $\relstr{A}$ can contain, in addition, elements of the domain of $\relstr{A}$. 

A \emph{pp-interpretation} is a first-order interpretation in the sense of model theory where all the involved formulas are primitive positive: 
a structure $\relstr{A}$ \emph{pp-interprets} $\relstr{B}$ (with parameters) if there exists a partial mapping $f$ from a finite power $A^n$ to $B$ such that
the domain of $f$, the $f$-preimage of the equality relation and the $f$-preimage of every relation in $\relstr{B}$ is pp-definable in $\relstr{A}$ (with parameters). 
In particular, $\relstr{A}$ pp-interprets its substructures induced by pp-definable subsets and also its quotients modulo a pp-definable equivalence relation. 
\lib{The number $n$ is refered to as the \emph{dimension} of the interpretation.}

\subsection{Cores} An $\omega$-categorical structure $\relstr{A}$ is a \emph{core}, also called \emph{model-complete core}, if all of its endomorphisms are elementary self-embeddings, i.e., preserve all first-order formulas over $\relstr{A}$. This is the case if and only if \lib{for each finite set $F \subseteq A$ each endomorphism of $\relstr{A}$ coincides with an automorphism of $\relstr{A}$ on the set $F$. 
In topological terms,} the automorphism group of $\relstr{A}$ is dense in its endomorphism monoid with respect to the pointwise convergence topology on functions on $A$; cf.~Section~\ref{sect:top} for a description of the latter.
Two structures $\relstr{A}, \relstr{B}$ are \emph{homomorphically equivalent} if there exist homomorphisms from $\relstr{A}$ into $\relstr{B}$ and vice-versa.

\subsection{Clones} 
A \emph{function clone} $\clone{C}$ is a set of finitary operations (also called functions) on a fixed set $C$ which contains all projections and which is closed under composition. 
A \emph{polymorphism} of  a relational structure $\relstr{A}$ is a finitary operation $f(x_1,\ldots,x_n)$ on $A$  which \emph{preserves} all relations $R$ of $\relstr{A}$: this means that for all $\tuple{r}_1,\ldots,\tuple{r}_n\in R$ we have that $f(\tuple{r}_1,\ldots,\tuple{r}_n)$, calculated componentwise, is again contained in $R$. The \emph{polymorphism clone} of $\relstr{A}$, denoted by $\Pol(\relstr{A})$, consists of all polymorphisms of $\relstr{A}$, and is always a function clone. Its unary operations are precisely the endomorphisms of $\relstr{A}$, and its \emph{invertible} unary operations (i.e., those which are bijections and whose inverse is also a polymorphism) are precisely the automorphisms of $\relstr{A}$.

A \emph{clone homomorphism} is a mapping from one function clone to another which preserves arities, composition, and which sends every projection of its domain to the corresponding projection of its co-domain. Equivalently, a clone homomorphism is a mapping $\xi$ that preserves arities and \emph{identities}, i.e., universally quantified equations over $\clone{C}$: more precisely, whenever  an identity $t \approx s$ holds in $\clone{C}$, where $t$ and $s$ are terms over the signature which has one functional symbol for every element in $\clone{C}$, then the identity obtained from $t \approx s$  by replacing each \co{occurrence} of $f \in \clone{C}$ by $\xi(f)$ holds in $\clone{D}$. See~\cite{wonderland} for a more detailed exposition of connections between various variants of clone homomorphisms and identities.

\subsection{Topology}\label{sect:top}

Function clones carry a natural topology, the topology of \emph{pointwise convergence}, for which a subbasis is given by sets of functions which agree on a fixed finite tuple; the functions of a fixed arity in a function clone form a clopen set. Equivalently, the domain of a function clone is taken to be discrete, and the $n$-ary functions in the clone equipped with the product topology, for every $n\geq 1$; the whole clone is then the sum space of the spaces of $n$-ary functions. The function clones which are topologically closed within the space of all functions on their domain are precisely the polymorphism clones of relational structures.

We always understand continuity of clone homomorphisms with respect to this topology.

\subsection{Core clones and oligomorphicity}

We say that a closed function clone is a \emph{core} 
if  it is the polymorphism clone of a core. 

A function clone $\clone{C}$ is \emph{oligomorphic} if and only if  the permutation group $\clgr{\clone{C}}$ of unary invertible elements of $\clone{C}$ is oligomorphic. When $\clone{C}$ is closed, then this is the case if and only if it is the polymorphism clone of an $\omega$-categorical structure, in which case $\clgr{\clone{C}}$ consists of the automorphisms of that structure. 

When $\clone{C}$ is a core, then the set of its unary operations is the closure of $\clgr{\clone{C}}$ \co{in the space of all finitary operations on its domain.}
 

\subsection{Pseudo-Siggers operations} A 6-ary operation $s$ in a function clone $\clone{C}$ is a \emph{pseudo-Siggers operation} if there exist unary $\alpha, \beta \in \clone{C}$ 
such that the identity $\alpha s(x,y,x,z,y,z) \approx \beta s(y,x,z,x,z,y)$ holds in $\clone{C}$  (i.e., equality holds for all values for the variables in $C$). We then also say that $s$ \emph{satisfies} the pseudo-Siggers identity.

\section{The Pseudoloop Lemma}\label{sect:loop}

The following definition is a generalization of finite graphs to the $\omega$-categorical which is suitable for our purposes.

\begin{definition}\label{defn:gg}
A \emph{graph-group-system}, for  short \emph{gg-system}, is a pair $(\relstr{G},\group{G})$, where $\group{G}$ is a permutation group on a set $G$, and $\relstr{G}=(G;R)$ an (undirected) graph which is invariant under $\group{G}$. We also write $(R,\group{G})$ for the same gg-system. 

The system is called \emph{oligomorphic} if $\group{G}$ is; in that case, $\relstr{G}$ is $\omega$-categorical, since its automorphism group contains $\group{G}$ and hence is oligomorphic. 

The system \emph{pp-interprets (pp-defines)} a structure $\relstr{B}$ if $\relstr{G}$ together with the orbits of $\group{G}$ on finite tuples does.

A \emph{pseudoloop} of a gg-system $(\relstr{G},\group{G})$ is an edge of $\relstr{G}$ of the form $(a,\alpha(a))$, where $\alpha\in\group{G}$. 
\end{definition}

Note that $R$, as well as any relation that is first-order definable from a gg-system $(\relstr{G},\group{G})$, is invariant under the natural action of $\group{G}$ on tuples. In particular, such relations are unions of orbits of the action of $\group{G}$ on tuples, and 
when $\group{G}$ is oligomorphic, then there are only finite many first-order definable relations of any fixed arity. 

We are now ready to state our pseudoloop lemma for gg-systems.

\begin{lemma}[The pseudoloop lemma]\label{lem:loop}
Let $(\relstr{G},\group{G})$ be an oligomorphic gg-system, where $\relstr{G}$ has a subgraph isomorphic to $\relstr{K}_3$. Then either it {pp}-interprets $\relstr{K}_3$ with parameters, or it contains a pseudoloop.
\end{lemma}

We need the following auxiliary notation and definitions.

\begin{definition}
Let $(\relstr{G},\group{G})$ be a gg-system. 
For $a_1,\ldots,a_n \in G$, we denote by $O(a_1,\ldots,a_n)$ the orbit of the tuple $(a_1,\ldots,a_n)$ under $\group{G}$.
\end{definition}


\begin{definition}
A pseudoloop-free gg-system $(\relstr{G},\group{G})$ \co{such} that $\relstr{G}$ contains (an isomorphic copy of) $\relstr{K}_3$ is \emph{minimal} if 
\begin{itemize}
\item it does not pp-define any proper subset $S\subsetneq G$ and a binary symmetric relation $R'$ on $S$ such that $R'$ is pseudoloop-free and contains a $\relstr{K}_3$, and
\item it does not pp-define any non-trivial equivalence relation $\sim\subseteq G^2$ and a binary symmetric relation $R'$ on the set $G/\sim$ of its equivalence classes such that $R'$ is pseudoloop-free and contains a $\relstr{K}_3$.
\end{itemize}
\end{definition}

We can now prove the pseudoloop lemma. \lib{Steps 0 to 4, 6, and 7 in the proof roughly follow the reasoning in~\cite{B05} with some adjustments to our setting. The main novelty is in Step 5, where  a new combinatorial argument is required.}

\begin{proof}[Proof of Lemma~\ref{lem:loop}]
Assuming that a gg-system $(\relstr{G},\group{G})$, where $\relstr{G} = (G;R)$ contains a $\relstr{K}_3$, has no pseudoloop, we show that it {pp}-interprets $\relstr{K}_3$ with parameters. \bigskip

\noindent \textit{Step 0}: If $(\relstr{G},\group{G})$ is not minimal, then we can replace it by a minimal gg-system as follows. When $(\relstr{G},\group{G})$ pp-defines a relation $R'$ on a pp-definable subset $S\subsetneq G$ which contains no pseudoloop but a $\relstr{K}_3$, then we replace $(\relstr{G},\group{G})$ by the system thus obtained on $S$, consisting of $R'$ and the restriction of the action of $\group{G}$ to $S$. If this is not the case, but  $(\relstr{G},\group{G})$ pp-defines a relation on a proper pp-definable factor of $G$ which contains no pseudoloop but a $\relstr{K}_3$, then  we proceed similarly, replacing the original system by the system obtained on the factor set. Iterating this procedure, it can only happen a finite number of times that we move to a system on a subset, since this step strictly decreases the finite number of orbits of the system, and since in factoring steps the number of orbits of the system does not increase either. Hence, from some point in the iteration, only factoring steps occur. But two factoring steps could be performed in a single step by combining the pp-definitions, and thus, since there exists only a finite number of pp-definable equivalence relations in the system, all factoring steps can be combined to a single step, after which the system cannot be further factored. The gg-system thus obtained must therefore be minimal.

Notice that minimality implies that every vertex is contained in an edge: otherwise restriction of $R$ to the (pp-definable) subset of those vertices which are contained in an edge would yield a smaller system.
\bigskip


\noindent \textit{Step 1}: $R$ pp-defines a symmetric binary relation $R'$ with the property that every edge of $R'$ is contained in a $\relstr{K}_3$, i.e.,  every element of $R'$ is contained in an induced  subgraph of $(G;R')$ isomorphic to $\relstr{K}_3$, and which still shares our assumptions on $R$: 
$$
R'(x,y):\leftrightarrow \exists z\; R(x,y)\wedge R(x,z)\wedge R(y,z).
$$
Hence, replacing $R$ by $R'$, we henceforth assume that every edge of $R$ is contained in a $\relstr{K}_3$.\bigskip

\noindent 
In the following, for $n\geq 1$ we say that $x,y\in G$ are \emph{$n$-diamond-connected}, denoted by $x\sim_n y$, if there exist $a_1$, $b_1$, $c_1$, $d_1$, $\ldots$, $a_n$, $b_n$, $c_n$, $d_n\in G$ such that, for every $1\leq i\leq n$, both $a_i,b_i,c_i$ and $b_i,c_i,d_i$ induce $\relstr{K}_3$ in $\relstr{G}$, 
$x=a_1$, $d_1=a_2$, $d_2=a_3$, \dots, $d_{n-1} = a_{n}$, and $d_n=y$. They are \emph{diamond-connected}, denoted by $x\sim y$, if they are $n$-diamond-connected for some $n\geq 1$. 

\MyFig{
  \begin{center}
    \begin{tikzpicture}[scale=1.2]
		 \node[ikrouzek] (a1) at (0,0) {$x=a_1$};
		 \node[krouzek] (b1) at (1,1) {$b_1$};
		 \node[krouzek] (c1) at (1,-1) {$c_1$};
		 \node[krouzek] (d1) at (2,0) {$d_1=a_2$};
		 \draw[myedge] (a1) to (b1); \draw[myedge] (a1) to (c1); \draw[myedge] (b1) to (c1);   \draw[myedge] (b1) to (d1); \draw[myedge] (c1) to (d1);
		 \node[krouzek] (b2) at (3,1) {$b_2$};
		 \node[krouzek] (c2) at (3,-1) {$c_2$};
		 \node[krouzek] (d2) at (4,0) {$d_2=a_3$};
		 \draw[myedge] (d1) to (b2); \draw[myedge] (d1) to (c2); \draw[myedge] (b2) to (c2);   \draw[myedge] (b2) to (d2); \draw[myedge] (c2) to (d2);
		 \draw[myedge] (d2) to (4.7,0.7); \draw[myedge] (d2) to (4.7,-0.7);
		 \node at (5.5,0) {$\cdots$};
		 \node[krouzek] (dn1) at (7,0) {$d_{n-1}=a_n$};
		 \draw[myedge]  (6.3,0.7) to (dn1); \draw[myedge]  (6.3,-0.7) to (dn1);
		 \node[krouzek] (bn) at (8,1) {$b_n$};
		 \node[krouzek] (cn) at (8,-1) {$c_n$};
		 \node[ikrouzek] (dn) at (9,0) {$d_n=y$};
		 \draw[myedge] (dn1) to (bn); \draw[myedge] (dn1) to (cn); \draw[myedge] (bn) to (cn);   \draw[myedge] (bn) to (dn); \draw[myedge] (cn) to (dn);
  \end{tikzpicture}
  \end{center}
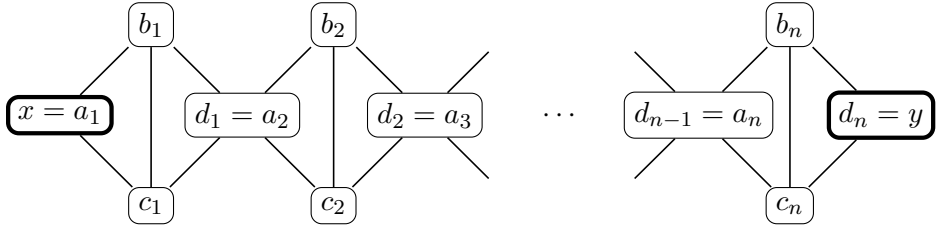
  \caption{Diamond connected elements.}
  \label{fig:dimond_con}
}

Observe that $\sim_n$  is a pp-definable relation from $R$ (since our definition is in fact a pp-definition). Also recall that there are only finitely many binary relations first-order definable from $R$, and note that if $x,y$ are $n$-diamond-connected, then they are $m$-diamond-connected for all $m\geq n$. Therefore, there exists an $n\geq 1$ such that $x,y$ are diamond-connected if and only if they are $n$-diamond connected. In particular, the relation $x\sim y$ is pp-definable in $\relstr{G}$. Note also that it is an equivalence relation on $G$: it is clearly transitive and symmetric, and it is reflexive since every vertex is contained in a $\relstr{K}_3$, by Steps~0 and~1.\bigskip

\noindent \textit{Step 2}: We claim that if $x, y\in G$ are $n$-diamond-connected for some $n \geq 1$, then $\neg R(x,y')$ for all $y'\in O(y)$. Otherwise, pick a counterexample $x,y,y'$ with minimal $n\geq 1$. 

Suppose first that $n$ is odd and set $k := \frac{n-1}{2}$. Let $a$ be the $a_{k+1}$ from the chain of diamonds witnessing $x\sim_n y$. 
Consider the following pp-definition over $(\relstr{G},\group{G})$:
$$
S(w):\leftrightarrow \exists u,v\;(u\in O(a)\wedge u\sim_{k} v\wedge R(v,w))\; ;
$$
in case that $k=0$ we replace $\sim_k$ by the equality relation. Then clearly $S(b_n)$ and $S(c_n)$. But we also have $S(y)$, since $S(y')$ holds by virtue of $a\sim_k x$ and $R(x,y')$ and since $y$ is in the same orbit as $y'$. 
Hence, since $d_n=y$, we have $S(d_n)$ and so $S$ contains a $\relstr{K}_3$. 
By the minimality of $(\relstr{G},\group{G})$ (see Step~0), we must have that $S$ holds for all elements of $G$. 
\MyFig{
  \begin{center}
    \begin{tikzpicture}[scale=1]
         \node[krouzek] (yprime) at (-1,0) {$y'$}; 
		 \node[krouzek] (a1) at (0,0) {$x$};
		 \draw[myedge] (yprime) to (a1);
		 \node[krouzek] (b1) at (1,0.7) {$\ $};
		 \node[krouzek] (c1) at (1,-0.7) {$\ $};
		 \node[krouzek] (d1) at (2,0) {$a$};
		 \draw[myedge] (a1) to (b1); \draw[myedge] (a1) to (c1); \draw[myedge] (b1) to (c1);   \draw[myedge] (b1) to (d1); \draw[myedge] (c1) to (d1);
		 \node[krouzek] (b2) at (3,0.7) {$\ $};
		 \node[krouzek] (c2) at (3,-0.7) {$\ $};
		 \node[krouzek] (d2) at (4,0) {$\ $};
		 \draw[myedge] (d1) to (b2); \draw[myedge] (d1) to (c2); \draw[myedge] (b2) to (c2);   \draw[myedge] (b2) to (d2); \draw[myedge] (c2) to (d2);
		 \node[krouzek] (b3) at (5,0.7) {$b_3$};
		 \node[krouzek] (c3) at (5,-0.7) {$c_3$};
		 \node[krouzek] (d3) at (6,0) {$y$};
		 \draw[myedge] (d2) to (b3); \draw[myedge] (d2) to (c3); \draw[myedge] (b3) to (c3);   \draw[myedge] (b3) to (d3); \draw[myedge] (c3) to (d3);
		 \node[krouzek] (u) at (2,-1.9) {$u$};
		 \node[krouzek] (bu) at (3,-1.2) {$\ $};
		 \node[krouzek] (cu) at (3,-2.6) {$\ $};
		 \node[krouzek] (v) at (4,-1.9) {$v$};
		 \draw[myedge] (u) to (bu); \draw[myedge] (u) to (cu); \draw[myedge] (bu) to (cu);   \draw[myedge] (bu) to (v); \draw[myedge] (cu) to (v);
		 \node[ikrouzek] (w) at (5,-1.9) {$w$}; \draw[myedge] (v) to (w);
		 \draw[myedge,dotted] (d1) to[bend right] (u);
  \end{tikzpicture}
  \end{center}
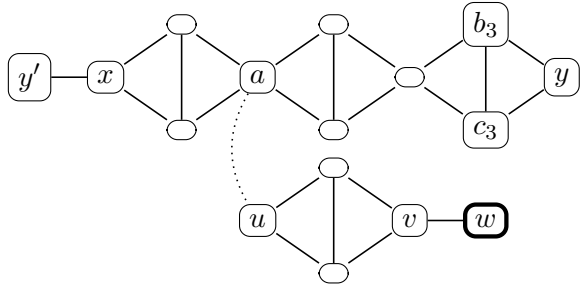
  \caption{The case $n=3$ and the definition of $S$. The dotted line depicts the ``same orbit'' relation.}
  \label{fig:dimond_con_odd}
}
Let $u, v\in G$ as in the definition of $S$ witness that $S(x)$ holds. Then $u\sim _{k} v$, but also $u\sim_{k} x'$ for some $x'\in O(x)$, as $a \sim_k x$, $u \in O(a)$, and $\sim_k$ is invariant under $\group{G}$. Therefore, $v\sim_{n-1} x'$, which together with  $R(v,x)$ contradicts the minimality of $n$ when $n\geq 3$; when $n=1$, this means that we have discovered a pseudoloop of $(\relstr{G},\group{G})$, again a contradiction.

Suppose now that $n$ is even and denote $k := \frac{n}{2}-1$; the argument is similar. 
Let $b,c$ be the $b_{k+1}, c_{k+1}$ from the chain of diamonds witnessing $x\sim_n y$. 
Consider the following pp-definition:
\begin{align*}
S(w):\leftrightarrow \exists & u_b,u_c,u,v\;((u_b,u_c)\in O(b,c)\wedge \\
  & R(u,u_c)\wedge R(u,u_b)\wedge u\sim_{k} v\wedge R(v,w))\; .
\end{align*}
\MyFig{
  \begin{center}
    \begin{tikzpicture}[scale=1]
         \node[krouzek] (yprime) at (-1,0) {$y'$}; 
		 \node[krouzek] (a1) at (0,0) {$x$};
		 \draw[myedge] (yprime) to (a1);
		 \node[krouzek] (b1) at (1,0.7) {$\ $};
		 \node[krouzek] (c1) at (1,-0.7) {$\ $};
		 \node[krouzek] (d1) at (2,0) {$\ $};
		 \draw[myedge] (a1) to (b1); \draw[myedge] (a1) to (c1); \draw[myedge] (b1) to (c1);   \draw[myedge] (b1) to (d1); \draw[myedge] (c1) to (d1);
		 \node[krouzek] (b2) at (3,0.7) {$b$};
		 \node[krouzek] (c2) at (3,-0.7) {$c$};
		 \node[krouzek] (d2) at (4,0) {$\ $};
		 \draw[myedge] (d1) to (b2); \draw[myedge] (d1) to (c2); \draw[myedge] (b2) to (c2);   \draw[myedge] (b2) to (d2); \draw[myedge] (c2) to (d2);
		 \node[krouzek] (b3) at (5,0.7) {$\ $};
		 \node[krouzek] (c3) at (5,-0.7) {$\ $};
		 \node[krouzek] (d3) at (6,0) {$\ $};
		 \draw[myedge] (d2) to (b3); \draw[myedge] (d2) to (c3); \draw[myedge] (b3) to (c3);   \draw[myedge] (b3) to (d3); \draw[myedge] (c3) to (d3);
		 \node[krouzek] (b4) at (7,0.7) {$b_4$};
		 \node[krouzek] (c4) at (7,-0.7) {$c_4$};
		 \node[krouzek] (d4) at (8,0) {$y$};
		 \draw[myedge] (d3) to (b4); \draw[myedge] (d3) to (c4); \draw[myedge] (b4) to (c4);   \draw[myedge] (b4) to (d4); \draw[myedge] (c4) to (d4);
		 \node[krouzek] (bu) at (3,-1.3) {$u_b$};
		 \node[krouzek] (cu) at (3,-2.7) {$u_c$};
		 \node[krouzek] (v) at (4,-2) {$\ $};
		 \draw[myedge] (bu) to (cu);   \draw[myedge] (bu) to (v); \draw[myedge] (cu) to (v);
		 \node[krouzek] (bx) at (5,-1.3) {$\ $};
		 \node[krouzek] (cx) at (5,-2.7) {$\ $};
		 \node[krouzek] (dx) at (6,-2) {$v$};
		 \draw[myedge] (v) to (bx); \draw[myedge] (v) to (cx); \draw[myedge] (bx) to (cx);   \draw[myedge] (bx) to (dx); \draw[myedge] (cx) to (dx);
		 		 \node[ikrouzek] (w) at (7,-2) {$w$}; \draw[myedge] (dx) to (w);
		 \draw[myedge,dotted] (2.9,0) to[bend right,looseness=1.4] (2.9,-2);
  \end{tikzpicture}
  \end{center}
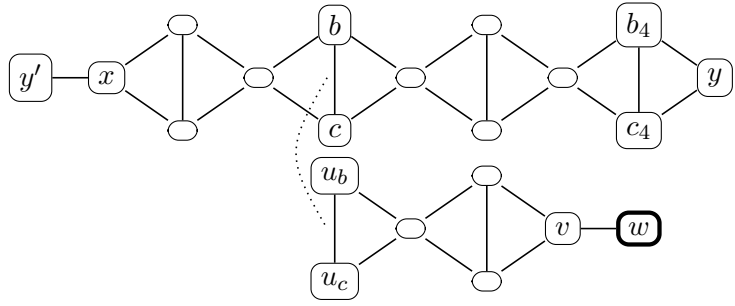
  \caption{The case $n=4$ and the definition of $S$. The dotted line depicts the ``same orbit'' relation.}
  \label{fig:dimond_con_even}
}
Then as in the odd case, $S(b_n), S(c_n), S(d_n)$, and so the set defined by $S$ contains a $\relstr{K}_3$. 
By the minimality of $(\relstr{G},\group{G})$, it contains $x$; let $u_b,u_c,u,v\in G$ as in the definition of $S$ witness this. \co{Then $u\sim _{k} v$ by the definition of $S$. Moreover,  
 $u\sim _{k+1} x'$ for some $x'\in O(x)$; this is the case since $u\sim_{1} d_{k}'$ for some $d_{k}'\in O(d_{k})$, and since $d_{k}'\sim_{k} x'$ for some $x'\in O(x)$.} Hence, $v\sim_{n-1} x'$ and $R(v,x)$ contradict the minimality of $n$.\bigskip
 
 \noindent \textit{Step 3}: Defining 
 $$
 R'(x,y):\leftrightarrow \exists x', y'\; (x\sim x'\wedge y\sim y'\wedge R(x',y'))
 $$
 we obtain a relation $R'\supseteq R$ which does not contain a pseudoloop. Indeed,
if $R'(x,y)$ is witnessed by $x',y'$ and $x$ and $y$ are in the same orbit, then 
$x \sim y''$ for some $y'' \in O(y')$ since $y\in O(x)$ and since $\sim$ is invariant under $\group{G}$. Thus $x' \sim y''$ and $R(x',y')$, a contradiction with Step~2.
Moreover, every edge in $R'$ is contained in a $\relstr{K}_3$: if $z'$ is so that $\{x',y',z'\}$ induce a $\relstr{K}_3$ in $R$, 
then $\{x,y,z'\}$ induce a $\relstr{K}_3$ in $R'$, for $z'\sim z'$ and $x\sim x'$ imply $R'(x,z')$, and $R'(y,z')$ can be inferred  similarly. 
\bigskip

\noindent \textit{Step 4}: Note that $R'$ is in fact a relation between equivalence classes of $\sim$ and the naturally defined quotient gg-system $(\relstr{G}^q,\group{G}^q)$ on $G^q = G/\sim$ contains no pseudoloops (by Step~3). 
Moreover, every edge of $\relstr{G}^q$ is still contained in a $\relstr{K}_3$. By minimality, $\sim$ must therefore be the equality relation. This means that $R$ contains no \emph{diamonds}, that is, there do not exist distinct $a,b,c,d \in G$ such that $\{a,b,c\}$ and $\{b,c,d\}$ both induce a $\relstr{K}_3$.

Summarizing we now know that our gg-system $(\relstr{G},\group{G})$ is minimal, pseudoloop-free, and diamond-free; moreover, every edge and every vertex  of $\relstr{G}$ is contained in a $\relstr{K}_3$. 
\bigskip

\noindent \textit{Step 5}: For $k\geq 1$, we denote the $k$-th power of $\relstr{K}_3$ by $\relstr{T}_k$.
By Lemma~\ref{lem:Tk_bound} shown below, $\relstr{G}$ contains no induced subgraph isomorphic to $\relstr{T}_k$ for any $k\geq 1$ such that $|T_k|$ exceeds the number of orbits of $\group{G}$.


\bigskip

\noindent \textit{Step 6}: Recall that $\relstr{G}$ contains $\relstr{T}_1 = \relstr{K}_3$. By Step 5, there exists a maximal $k\geq 1$ such that $\relstr{G}$ contains an induced subgraph isomorphic to $\relstr{T}_k$.
Let $k$ be that number and let $a_1, \dots, a_l$, where $l := |T_k| = 3^{k}$, denote the vertices of such an induced subgraph. 
We show that $\relstr{G}$ pp-defines the set $A = \{a_1, \dots, a_l\}$ with parameters $a_1,\ldots,a_l$. 

By~\cite{BodirskyNesetrilJLC}, this is the case if each $l$-ary operation $f$ in $\Pol(\relstr{G},a_1, \dots, a_l)$ preserves $A$. So, suppose that such a function $f$ does not preserve $A$.
Now, $f$ is a homomorphism $\relstr{G}^l \To \relstr{G}$ and its restriction to \co{$A^l$} is a homomorphism $f'$ from \co{$(\relstr{T}_k)^l$} to the diamond-free graph $\relstr{G}$ whose image, which contains $A$ because $f$ stabilizes each $a_i$, is strictly larger than $|T_k|$. Then~\cite[Claim 3, Subsection 3.2]{B05} shows that the image of $f'$ induces a graph isomorphic to $\relstr{T}_m$ for some $m > k$, contradicting the maximality of $k$.\bigskip

\noindent \textit{Step 7}: Step 6 implies that $\relstr{G}$ pp-interprets $\relstr{T}_k$ with parameters. But $\relstr{K}_3$ can be pp-interpreted in $\relstr{T}_k$ with parameters by the final sentence of~\cite{B05}.
\end{proof}


It remains to prove the lemma referred to in Step~5 of the proof of Lemma~\ref{lem:loop} above. It might be worth noting  that the essential conclusion of Step~5, namely that $\relstr{G}$ can be assumed not to contain an isomorphic copy of $\relstr{T}_k$ for $k$ large enough, is trivial in the proof of the finite loop lemma in~\cite{B05} for cardinality reasons: it is certainly true for all $k$ such that $|T_k|>|G|$. Lifting this statement to the $\omega$-categorical setting seems to require a non-trivial combinatorial argument rather than, for example, a simple compactness argument. This is accomplished in the following lemma.

\begin{lemma}\label{lem:Tk_bound}
Let $(\relstr{G},\group{G})$, where 
$\relstr{G}=(G;R)$, 
be a minimal (hence pseudoloop-free)  oligomorphic gg-system. Assume moreover that  
 every edge of $\relstr{G}$ is contained in a $\relstr{K}_3$. Then $\relstr{G}$ does not contain any induced subgraph isomorphic to $\relstr{T}_k$ where $k\geq 1$ is such that  $|T_k|=3^{k}$  exceeds the number of orbits of $\group{G}$.

\ignore{
Let $(\relstr{G},\group{G})$, where 
$\relstr{G}=(G;R)$, 
be a minimal (hence pseudoloop-free)  oligomorphic gg-system. Assume moreover that  
\begin{itemize}
\item every edge of $\relstr{G}$ is contained in a $\relstr{K}_3$, and 
\item $\relstr{G}$ contains an induced subgraph isomorphic to $\relstr{T}_k$ for some $k\geq 1$ such that  $|T_k|=3^{k}$  exceeds the number of orbits of $\group{G}$.
\end{itemize}
Then $(\relstr{G},\group{G})$ 
pp-defines a symmetric relation $R'\supsetneq R$ without pseudoloops such that every edge of $R'$ is contained in a $\relstr{K}_3$.\bigskip
}
\end{lemma}
\begin{proof}
Assume to the contrary that the lemma does not hold, and fix a copy of $\relstr{T}_k$ in $\relstr{G}$, the elements of which we denote by tuples in $\{1,2,3\}^k$. So, two vertices in $\{1,2,3\}^k$ are adjacent if and only if they differ in every coordinate. 
From the cardinality assumption, we can pick two elements $\tuple{a},\tuple{a}'$ of the copy that belong to the same orbit $A$. Let $\tuple{b},\tuple{c}$ in the copy be so that $\{\tuple{a},\tuple{b},\tuple{c}\}$ induce a $\relstr{K}_3$, and let $B,C$ be their orbits. Since $(\relstr{G},\group{G})$ has no pseudoloops, the three orbits $A,B,C$ are distinct. Without loss of generality, assume $\tuple{a}=1^k$ (i.e., the $k$ tuple all of whose entries equal $1$), $\tuple{b}=2^k$, and $\tuple{c}=3^k$.

 Define a relation
\begin{align*}
	S(u,v):\leftrightarrow \exists & a'',b'',c'', n_A,n_B,n_C \; \\
	&(R(u,n_A)\wedge R(v,n_A)\wedge R(n_A,a'')\wedge a''\in A\; \wedge\\
	& R(u,n_B)\wedge R(v,n_B)\wedge R(n_B,b'')\wedge b''\in B\; \wedge\\
	& R(u,n_C)\wedge R(v,n_C)\wedge R(n_C,c'')\wedge c''\in C)\; .
\end{align*}
In words, $u,v$ have common neighbors adjacent to elements in $A, B,$ and $C$. 

\MyFig{
  \begin{center}
    \begin{tikzpicture}[scale=1]
         \node[ikrouzek] (u) at (0,2) {$u$}; 
		 \node[ikrouzek] (v) at (0,1) {$v$};
		 \node[krouzek] (nc) at (1.5,0.5) {$n_C$};
		 \node[krouzek] (nb) at (1.5,1.5) {$n_B$};
		 \node[krouzek] (na) at (1.5,2.5) {$n_A$};
		 \node[krouzek] (c) at (3,0.5) {$c'' \in C$};
		 \node[krouzek] (b) at (3,1.5) {$b'' \in B$};
		 \node[krouzek] (a) at (3,2.5) {$a'' \in A$};
		 \draw[myedge] (u) to (na); \draw[myedge] (u) to (nb); \draw[myedge] (u) to (nc);
		 \draw[myedge] (v) to (na); \draw[myedge] (v) to (nb); \draw[myedge] (v) to (nc);
		 \draw[myedge] (a) to (na); \draw[myedge] (b) to (nb); \draw[myedge] (c) to (nc);
		 \node[ikrouzek] (u2) at (6,2.5) {$u$};
		 \node[ikrouzek] (v2) at (8,0.5) {$v$};
		 \node[krouzek] (s) at (8,2.5) {$s$};
		 \node[krouzek] (t) at (6,0.5) {$t$};
		 \draw[myedge] (u2) to node[above] {$R$} (s);
		 \draw[myedge] (t) to node[above] {$R$} (v2);
		 \draw[myedge] (s) to node[left] {$S$} (v2);
		 \draw[myedge] (t) to node[left] {$S$} (u2);
  \end{tikzpicture}
  \end{center}
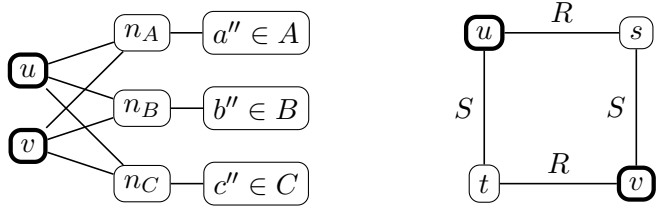
  \caption{The definition of $S$ and $Q$.}
  \label{fig:def_SQ}
}

The relation $S$ is obviously symmetric. It is also reflexive. To see this, it suffices to observe that 
 every element in $G$ 
 is a neighbor of a neighbor of an element in $A$, and similarly in $B$ and $C$. But the latter follows from  the minimality of $(\relstr{G},\group{G})$: otherwise, we could restrict $R$ to neighbors of neighbors of $A$, a set which contains $A\cup B\cup C$, and thus  obtain a smaller gg-system containing a $\relstr{K}_3$, namely the one induced by $\{\tuple{a},\tuple{b},\tuple{c}\}$. 

Observe that whenever $S(u,v)$ holds, then every element of $G$ is adjacent to a common neighbor of $O(u)$ and $O(v)$: this follows as above from the minimality of $(\relstr{G},\group{G})$ since the elements of  $A\cup B \cup C$ are adjacent to a common neighbor of $O(u)$ and $O(v)$.

Set
$$
Q(u,v):\leftrightarrow \exists s\; (R(u,s)\wedge S(s,v))\; \wedge\; \exists t\; (S(u,t)\wedge R(t,v))\; .
$$
Then $Q\supseteq R$: since $S$ is reflexive, setting $s=v$ and $t=u$ in the above definition shows that $R(u,v)$ implies $Q(u,v)$. Moreover, $Q$ is symmetric by definition. Let $R'$ consist of those edges of $Q$ which are contained in a $\relstr{K}_3$ with respect to $Q$. We still have that $R'\supseteq R$.

We now show that $(Q,\group{G})$, and thus $(R',\group{G})$, has no pseudoloop. To this end, it suffices to show that whenever $R(u,v)$ holds, then we cannot have $S(u,v')$ for any $v'\in O(v)$. Suppose to the contrary that there exist such elements. 
The $R$-edge $(u,v)$ is contained in a $\relstr{K}_3$, induced by $\{u,v,w\}$, for some $w\in G$.
As observed above, each vertex, in particular the vertex $w$, is adjacent to a common neighbor of $O(u)$ and $O(v')=O(v)$. 
Therefore, there exists a common neighbor $z$ of $O(u)$, $O(v)$ and $O(w)$. 
The set of neighbors of $O(z)$ contains $O(u), O(v)$, and $O(w)$; it is a proper subset of $G$ since $(\relstr{G},\group{G})$ has no pseudoloops; it is pp-definable in $(\relstr{G},\group{G})$; and finally, it contains a $\relstr{K}_3$, contradicting the minimality of $(\relstr{G},\group{G})$.

We claim that $\tuple{a}'$, the second element of the copy of $\relstr{T}_k$ in the orbit $A$ of $\tuple{a}$, is related to $\tuple{b}$ and $\tuple{c}$ via $R'$. This implies that $R'$ contains a diamond as an induced subgraph, a contradiction since the minimality of the system $(\relstr{G},\group{G})$ does imply diamond-freeness of $R$, and similarly of $R'$, by Steps~1 to~4 of the proof of Lemma~\ref{lem:loop}.

To prove our claim, we only show  $R'(\tuple{a}',\tuple{b})$, the second claim is analogous. 
Reordering the tuples when necessary, we may assume that $a'_i\neq 2$ for all $1\leq i\leq j$, and $a'_i= 2$ for all $j<i\leq k$. Since $\tuple{a}'\neq \tuple{b}$, we have $j\geq 1$. 
\lib{
Observe that whenever $\tuple{u}, \tuple{v} \in \{1,2,3\}^k$ are of the form $(x_1,\ldots,x_j,2,\ldots,2)$ and $(x_1,\ldots,x_j,3,\ldots,3)$, respectively, then $S(\tuple{u},\tuple{v})$: this is witnessed by their common neighbor $(y_1,\ldots,y_j,1,\ldots,1)$, where $y_i\notin\{a'_i,x_i\}$ for all $1\leq i\leq j$, which is $R$-related to $\tuple{a}'\in A$; their common neighbor $(z_1,\ldots,z_j,1,\ldots,1)$, where  $z_i\notin\{2,x_i\}$, which is $R$-related to $\tuple{b}\in B$; and their common neighbor $(w_1,\ldots,w_j,1,\ldots,1)$, where $w_i\notin\{3,x_i\}$, which is $R$-related to $\tuple{c}\in C$. See Figure~\ref{fig:final_SQ}.}
But now we see that $Q(\tuple{a}',\tuple{b})$ holds: setting $\tuple{t} := (a_1',\ldots,a_j',3,\ldots,3)$, 
we have $S(\tuple{a}',\tuple{t})$ and $R(\tuple{t},\tuple{b})$; on the other hand, setting $\tuple{s}:=(2,\ldots,2,3,\ldots,3)$, with $j$ occurrences of $2$, we have $R(\tuple{a}',\tuple{s})$ and $S(\tuple{s},\tuple{b})$. \co{See again  Figure~\ref{fig:final_SQ}.} We can then conclude that $R'(\tuple{a}',\tuple{b})$ holds, since any two elements of $\{1,2,3\}^k$, in particular $\tuple{a}'$ and $\tuple{b}$, have a common neighbor with respect to $R$, and hence also with respect to $Q$, showing that the $Q$-edge $(\tuple{a}',\tuple{b})$ is contained in a $\relstr{K}_3$ with respect to $Q$. 
\MyFig{
  \begin{center}
    \begin{tikzpicture}[scale=1]
         \node[krouzek] (u) at (0,2) {$x_1\dots x_j2\dots 2$}; 
		 \node[krouzek] (v) at (0,1) {$x_1 \dots x_j3\dots 3$};
		 \node[krouzek] (na) at (3,2.5) {$y_1\dots y_j1 \dots 1$};
		 \node[krouzek] (nb) at (3,1.5) {$z_1\dots z_j1\dots 1$};
		 \node[krouzek] (nc) at (3,0.5) {$w_1 \dots w_j1 \dots 1$};
		 \node[krouzek] (a) at (6,2.5) {$a_1' \dots a_j'2\dots 2$};
		 \node[krouzek] (b) at (6,1.5) {$2\dots 2 2\dots 2$};
		 \node[krouzek] (c) at (6,0.5) {$3\dots 3 3\dots 3$};
		 \draw[myedge] (u) to (na); \draw[myedge] (u) to (nb); \draw[myedge] (u) to (nc);
		 \draw[myedge] (v) to (na); \draw[myedge] (v) to (nb); \draw[myedge] (v) to (nc);
		 \draw[myedge] (a) to (na); \draw[myedge] (b) to (nb); \draw[myedge] (c) to (nc);
 		 \node[krouzek] (u2) at (9,2.5) {$a_1'\dots a_j'2 \dots 2$};
		 \node[krouzek] (v2) at (12,0.5) {$2\dots2 2 \dots 2$};
		 \node[krouzek] (s) at (12,2.5) {$2 \dots 2 3 \dots 3$};
		 \node[krouzek] (t) at (9,0.5) {$a_1' \dots a_j' 3 \dots 3$}; 
		 \draw[myedge] (u2) to node[above] {$R$} (s);
		 \draw[myedge] (t) to node[above] {$R$} (v2);
		 \draw[myedge] (s) to node[left] {$S$} (v2);
		 \draw[myedge] (t) to node[left] {$S$} (u2);
		 \end{tikzpicture}
  \end{center}
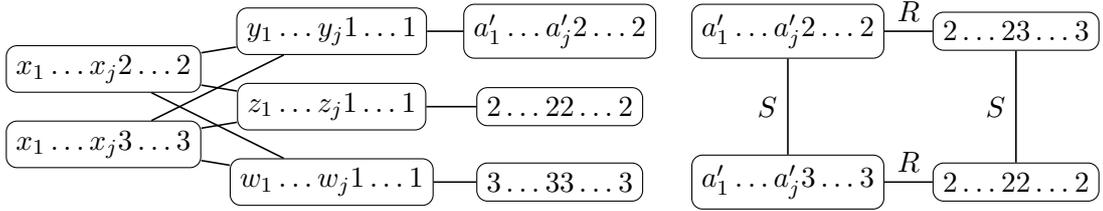
  \caption{Arguments that $S(x_1 \dots x_j 2 \dots 2,x_1 \dots x_j 3 \dots 3)$ and $Q(\tuple{a}',\tuple{b})$.}
  \label{fig:final_SQ}
}
\end{proof}

\section{The main result} \label{sect:main}



In order to derive Theorem~\ref{thm:main}, we will produce pseudo-Siggers operations locally using the pseudoloop lemma, and then derive a global pseudo-Siggers operation via a compactness argument.

\begin{definition}
We say that a function clone $\clone{C}$ has \emph{local pseudo-Siggers operations} if for every finite $A\subseteq C$ there exists a $6$-ary $s\in\clone{C}$ and unary $\alpha, \beta \in\clone{C}$ satisfying 
\[
\alpha s(x,y,x,z,y,z) = \beta s(y,x,z,x,z,y)
\]
for all $x,y,z \in A$.
\end{definition}

\begin{lemma}\label{lem:localglobal}
Let $\clone{C}$ be a closed oligomorphic function clone. If it has local pseudo-Siggers operations, then it has a pseudo-Siggers operation.
\end{lemma}
\begin{proof}
Let $A_0\subseteq A_1\subseteq\cdots$ be a sequence of finite subsets of $C$ whose union equals $C$, and pick for every $i\in \omega$ a $6$-ary operation $s_i\in\clone{C}$ witnessing the definition of local pseudo-Siggers operations on $A_i$, i.e., there exist unary $\alpha_i,\beta_i\in\clone{C}$ such that $\alpha_i s_i(x,y,x,z,y,z) = \beta_i s_i(y,x,z,x,z,y)$ for all $x,y,z\in A_i$. 
 Note that if  $s_i$ is such a witness for $A_i$, then so is $\gamma s_i$, for all $\gamma\in\clgr{\clone{C}}$. Hence, because $\clgr{\clone{C}}$ is oligomorphic,  we may thin out the sequence in such a way that $s_j$ agrees with $s_i$ on $A_i$, for all $j>i\geq 0$. We briefly describe this standard compactness argument for the convenience of the reader: there exists a smallest $j_0\geq 0$ such that for infinitely many $k\geq j_0$ there exists $\gamma_k\in\clgr{\clone{C}}$ such that $\gamma_k s_k$ agrees with $s_{j_0}$ on $A_0$, by oligomorphicity. Replace $s_0$ by $s_{j_0}$, all $s_k$ as above by $\gamma_k s_k$, and remove all other $s_{k'}$ where $k'\geq 0$ from the sequence. Next repeat this process picking $j_1\geq 1$ for $A_1$, and so on. This completes the argument.
 
Since the elements of the sequence $(s_i)_{i\in\omega}$ agree on every fixed $A_i$ eventually, and since $\clone{C}$ is closed, they converge to a function $s\in \clone{C}$. The function $s$, restricted to any $A_i$, witnesses local pseudo-Siggers operations on $A_i$, i.e., there exist unary $\alpha_i,\beta_i\in\clone{C}$ such that $\alpha_i s(x,y,x,z,y,z) = \beta_i s(y,x,z,x,z,y)$ for all $x,y,z\in A_i$. By a similar compactness argument as the one above for the functions $s_i$, this time applied to the pairs $(\alpha_i,\beta_i)$ of functions, we may assume that these pairs converge to a pair $(\alpha,\beta)$ of functions $\alpha,\beta\in\clone{C}$. We then have  $\alpha s(x,y,x,z,y,z) = \beta s(y,x,z,x,z,y)$ for all $x,y,z\in C$.
\end{proof}

We now consider gg-systems where the group $\clgr{\clone{C}}$ of a closed oligomorphic function clone $\clone{C}$ acts on finite powers of its domain.

\begin{lemma}\label{lem:1}
Let $\clone{C}$ be a closed oligomorphic function clone. Suppose that every gg-system $(\relstr{G},\group{G})$ where
\begin{itemize}
\item $\relstr{G}=(C^k;R)$ for some $k\geq 1$,
\item $\group{G}$ corresponds to the componentwise action of $\clgr{\clone{C}}$ on $C^k$,
\item $\relstr{G}$ contains a  $\relstr{K}_3$, and 
\item $R\subseteq C^{2k}$ is invariant under $\clone{C}$
\end{itemize}
has a pseudoloop. Then $\clone{C}$ has a pseudo-Siggers operation.
\end{lemma}
\begin{proof}
We show that $\clone{C}$ has local pseudo-Siggers operations and apply Lemma~\ref{lem:localglobal}. Let $A\subseteq C$ be finite, and pick $k\geq 1$ and $\tuple{a}^x,\tuple{a}^y,\tuple{a}^z\in A^k$ such that the rows of the $(k\times 3)$-matrix $(\tuple{a}^x,\tuple{a}^y,\tuple{a}^z)$ form an enumeration of $A^3$. Let $R$ be the binary relation on $C^k$ where tuples $\tuple{b},\tuple{c}\in C^k$ are related via $R$ if there exists a $6$-ary $s\in \clone{C}$ such that $\tuple{b}=s(\tuple{a}^x,\tuple{a}^y,\tuple{a}^x,\tuple{a}^z,\tuple{a}^y,\tuple{a}^z)$ and $\tuple{c}=s(\tuple{a}^y,\tuple{a}^x,\tuple{a}^z,\tuple{a}^x,\tuple{a}^z,\tuple{a}^y)$ \co{($s$ is applied to the $k$-tuples componentwise)}. 
In other words, it is the $\clone{C}$-invariant  subset of $(2k)$-tuples generated by the six vectors obtained by concatenating $\tuple{a}^u$ and $\tuple{a}^v$, where $u,v\in\{x,y,z\}$ are distinct. The latter description reveals that $R$ is a symmetric relation on $C^k$ invariant under $\clone{C}$ and containing $\relstr{K}_3$, therefore the gg-system $(R,\group{G})$, where $\group{G}$ is the componentwise action of $\clgr{\clone{C}}$ on $C^k$,  has a pseudoloop $(\tuple{b},\tuple{c})$. That means that there exists a $6$-ary $s\in \clone{C}$ and $\alpha \in\clgr{\clone{C}}$ such that $s(\tuple{a}^x,\tuple{a}^y,\tuple{a}^x,\tuple{a}^z,\tuple{a}^y,\tuple{a}^z)=\alpha s(\tuple{a}^y,\tuple{a}^x,\tuple{a}^z,\tuple{a}^x,\tuple{a}^z,\tuple{a}^y)$, proving the claim.
\end{proof}

\begin{corollary}\label{cor:K3-Siggers}
Let $\relstr{A}$ be an $\omega$-categorical core. Then either it pp-interprets $\relstr{K}_3$ with parameters, or $\Pol(\relstr{A})$ has a pseudo-Siggers operation.
\end{corollary}
\begin{proof}
We apply Lemma~\ref{lem:1} to the clone $\clone{C}:=\Pol(\relstr{A})$; then $\clgr{\clone{C}}$ consists precisely of the automorphisms of $\relstr{A}$. 
If the assumptions of that lemma are satisfied, then $\clone{C}$ has a pseudo-Siggers operation.
Otherwise, there exists a pseudoloop-free gg-system $((C^k;R),\group{G})$ satisfying the four conditions. By Lemma~\ref{lem:loop}, this gg-system pp-interprets $\relstr{K}_3$ with parameters.  Since $R$ is invariant under $\clone{C}$, it is pp-definable from $\relstr{A}$ by~\cite{BodirskyNesetrilJLC}. Moreover, since $\relstr{A}$ is a core, the orbits of $\group{G}$ are pp-definable from $\relstr{A}$ as well by~\cite{Bodirsky-HDR}. It follows that $\relstr{A}$ pp-interprets $\relstr{K}_3$ with parameters, as required.
\end{proof}

We are now ready to prove Theorem~\ref{thm:main}. For the convenience of the reader, we restate it here.

\begin{theorem}[Theorem~\ref{thm:main}]
The following are equivalent for an $\omega$-categorical core structure $\relstr{A}$. 
\begin{itemize}
\item[(i)] There exists no continuous clone homomorphism $\Pol(\relstr{A},c_1,\ldots,$ $c_n) \To\trivclone$, for any $c_1,\ldots,c_n\in {A}$.
\item[(ii)] There exists no clone homomorphism $\Pol(\relstr{A},c_1,\ldots,c_n)\To\trivclone$, for any $c_1,\ldots,c_n\in {A}$.
\item[(iii)] $\Pol(\relstr{A})$ contains a pseudo-Siggers operation, i.e., a 6-ary operation $s$ such that 
\[
\alpha s(x,y,x,z,y,z) \approx \beta s(y,x,z,x,z,y)
\] 
for some unary operations $\alpha, \beta \in \Pol(\relstr{A})$.
\end{itemize}
\end{theorem}

\begin{proof}
We first prove that (iii) implies (ii).
Take $\alpha,\beta,s \in \Pol(\relstr{A})$ satisfying the pseudo-Siggers identity. We claim that every stabilizer $\Pol(\relstr{A},c_1,\dots, c_n)$ has a pseudo-Siggers operation. To see that, consider the endomorphisms $\gamma,\delta$ of $\relstr{A}$ defined by $\gamma(x):=s(x,\dots,x)$, $\delta(x) := \alpha\gamma(x)$ ($=\beta\gamma(x)$ by the pseudo-Siggers identity). Because $\relstr{A}$ is a core, its automorphisms are dense in endomorphisms, thus there exist automorphisms $\epsilon, \theta$ of $\relstr{A}$ such that $\epsilon(c_i) = \gamma(c_i)$ and $\theta(c_i) = \delta(c_i)$ for every $i$. But then $\theta^{-1}\alpha\epsilon$, $\theta^{-1}\beta\epsilon$ and $\epsilon^{-1}s$ are contained in $\Pol(\relstr{A},c_1, \dots, c_n)$ and satisfy 
\[
(\theta^{-1}\alpha\epsilon)(\epsilon^{-1}s)(x,y,x,z,y,z) \approx (\theta^{-1}\beta\epsilon)(\epsilon^{-1}s)(y,x,z,x,z,y). 
\]


The implication from (ii) to (i) is trivial. 

Finally, assume that no stabilizer of $\Pol(\relstr{A})$ has a continuous clone homomorphism to $\trivclone$. Then no such stabilizer has a continuous clone homomorphism to $\Pol(\relstr{K}_3)$ either, since it is well-known that the latter clone has a continuous clone homomorphism to $\trivclone$: 
 it is a folklore fact that all operations of $\Pol(\relstr{K}_3)$ are \emph{essentially unary}, i.e., depend on at most one variable. 
By Theorem~\ref{thm:int}, $\relstr{A}$ does not pp-interpret $\relstr{K}_3$ with parameters. Corollary~\ref{cor:K3-Siggers} then tells us that $\Pol(\relstr{A})$ has a pseudo-Siggers operation.
\end{proof}





\section{Discussion}\label{sect:discus}

\subsection{Discussion of the pseudoloop lemma}

The proof of the pseudoloop lemma presented here draws on   ideas from the proof in~\cite{B05} of the loop lemma for finite undirected graphs~\cite{HellNesetril,B05}, which states that every finite non-bipartite graph either contains a loop (i.e., an edge of the form $(a,a)$) or pp-interprets $\relstr{K}_3$ with parameters. A generalization of the latter result, which is useful in the theory of finite algebras and their applications to finite domain CSPs, 
 is its expansion from non-bipartite undirected graphs to certain directed graphs (digraphs). We shall now state this generalization.

Recall that a digraph is \emph{smooth} if each vertex has an incoming and an outgoing edge, and that a digraph has \emph{algebraic length 1} if it contains a closed walk whose number of forward edges exceeds its number of backward edges by $1$. Note that graphs containing an induced copy of $\relstr{K}_3$ have algebraic length $1$.

\begin{theorem}[\cite{BartoKozikNiven,Cyclic}] \label{thm:loop-lemma}
Let $\relstr{G} = (G;R)$ be a \co{finite} smooth  digraph of algebraic length $1$ without a loop. Then $\relstr{G}$  {pp}-interprets $\relstr{K}_3$ with parameters.
\end{theorem}

 We conjecture that the pseudoloop lemma can be generalized to such digraphs as well.

\begin{conjecture}\label{conj:loop}
Let $(\relstr{G},\group{G})$ be an oligomorphic gg-system, where $\relstr{G} = (G;R)$ is a smooth digraph of algebraic length $1$ without a pseudoloop. 
Then $(\relstr{G},\group{G})$ {pp}-interprets $\relstr{K}_3$ with parameters.
\end{conjecture}

It is easy to reduce the loop lemma for finite non-bipartite undirected graphs to finite undirected graphs containing a $\relstr{K}_3$ (see \cite{HellNesetril}). For the infinite pseudo-version, even this special case of Conjecture~\ref{conj:loop} is open. 
On the other hand, the generalization of the loop lemma for finite non-bipartite undirected graphs to the digraphs of  Theorem~\ref{thm:loop-lemma} requires completely different techniques. 
It might be possible to circumvent the additional complications in the infinite by reducing to the finite, as exemplified in the proof of the following proposition.

\begin{proposition}\label{prop:3orbits}
Let $(\relstr{G},\group{G})$ be an oligomorphic gg-system, where $\relstr{G} = (G;R)$ is a smooth connected digraph of algebraic length $1$.
Moreover, assume that 
\begin{itemize}
\item $\group{G}$ acts with three orbits on $G$, and
\item $\relstr{G}$ modulo the orbit equivalence on $G$ is  isomorphic to $\relstr{K}_3$ (i.e., when $U,V\subseteq G$ are orbits of $\group{G}$, then there is an edge from some $u\in U$ to some $v\in V$ if and only if $U\neq V$).
\end{itemize}
Then $(\relstr{G},\group{G})$ {pp}-interprets $\relstr{K}_3$ with parameters.
\end{proposition}

\begin{proof}
We prove the claim in three steps. In Step~1, we construct an equivalence relation $\sim$ on $G$ which is pp-definable from $(\relstr{G},\group{G})$ and has the following properties.
\begin{itemize}
\item[(i)] If $a \sim b$, then $O(a)=O(b)$.
\item[(ii)] If $a \sim b$, $R(a,a')$, $R(b,b')$, and $O(a')=O(b')$, then $a' \sim b'$.
\item[(iii)] If $a \sim b$, $R(a',a)$, $R(b',b)$, and $O(a')=O(b')$, then $a' \sim b'$.
\end{itemize}
In Step~2, we prove that $\relstr{G}$, factored by $\sim$, is a finite digraph. \co{This digraph satisfies the conditions of Theorem~\ref{thm:loop-lemma}, which we apply in Step~3 to finish the proof.} 

Before we start, we make a simple observation which will be used throughout the proof.
If $A,B\subseteq G$ are distinct orbits of $\group{G}$ and $a \in A$, then there exists $b \in B$ such that $R(a,b)$.
Indeed, since $\relstr{G}$ modulo the orbits is isomorphic to $\relstr{K}_3$, there are $a' \in A$ and $b' \in B$ such that $R(a',b')$, and so we can set \co{$b := \alpha(b')$} for any $\alpha \in \group{G}$ with $\alpha(a')=a$. A similar statement holds, of course, with the role of $a$ and $b$ reversed. 

\smallskip

\noindent \textit{Step 1:} For each pp-definable equivalence relation $\sim$ on $G$ satisfying (i) and elements $a,a',b,b'\in G$ falsifying (ii) (or, similarly, (iii)) we are going to pp-define an equivalence relation $\sim'$ containing $\sim$ which still satisfies (i), and such that  $a' \sim' b'$. This is enough for achieving the goal of Step~1 since the process of enlarging $\sim$ (starting with the smallest equivalence relation on $G$) must stop after a finite number of steps due to oligomorphicity (recall the remark after Definition~\ref{defn:gg}).

We will only consider the situation where (ii) does not yet hold for $\sim$, the situation for (iii) being completely analogous. So, let $\sim$ satisfy (i) and let $a,a',b,b'\in G$ be as in (ii). Denote $U:=O(a)=O(b)$ and $V:=O(a')=O(b')$,  and the remaining orbit of $\group{G}$ by $W$. 
Note that $U\cup V$ can be pp-defined since $x\in U\cup V$ if and only if $\exists y \; R(x,y) \wedge y \in W$ holds. Similarly, $U\cup W$ and $V\cup W$ are pp-definable in $(\relstr{G},\group{G})$.

We first construct a reflexive binary relation $S:=S_1 \cap S_2 \cap S_3 \cap S_4$ on $G$ which contains both $(a',b')$ and $(b',a')$ and satisfies (i). Using $S$ we define a symmetric and reflexive  relation $Q$ which contains $\sim \cup \{(a',b')\}$ and which still satisfies (i), and finally set $\sim'$ to be the transitive closure of $Q$.

The relations $S_i$ are pp-defined as follows:
\begin{align*}
S_1(x,y) &:\leftrightarrow \exists z_1,z_2 \; R(z_1,x) \wedge z_1 \sim z_2 \wedge R(z_2,y) \wedge z_1 \in U\cup V \\ 
S_2(x,y) &:\leftrightarrow \exists z_1,z_2 \; R(z_1,x) \wedge z_1 \sim z_2 \wedge R(z_2,y) \wedge z_1 \in U\cup W \\
S_3(x,y) &:\leftrightarrow \exists z_1, z_2, z_3, z_4 \; R(z_1,x) \wedge z_1 \sim z_2 \wedge R(z_2,z_3) \wedge R(z_4,z_3) \wedge R(z_4,y) \\
  & \quad \wedge z_1 \in U\cup V \wedge z_3 \in U\cup V \wedge z_4 \in V\cup W \\
S_4(x,y) &:\leftrightarrow S_3(y,x)\; .
\end{align*}
We claim that each $S_i$ is reflexive. Each $c \in G$ has an incoming $R$--edge from some $d \in U\cup V$. The choice $z_1=z_2:=d$ then shows $S_1(c,c)$, proving that $S_1$ is reflexive. Similarly, $S_2$ is reflexive as well. To prove the claim for $S_3$, let $c\in G$ arbitrary. If $c \in U\cup V$ then, for any $d \in U\cup V, e \in V\cup W$ with $R(d,c)$ and $R(e,c)$, the choice $(z_1,z_2,z_3,z_4) := (d,d,c,e)$ shows $S_3(c,c)$; if $c \in W$, then we pick $d \in V$, $e \in U$ such that $R(d,c)$, $R(d,e)$ and the choice $(z_1,z_2,z_3,z_4) := (d,d,e,d)$ works.

It is also easy to verify that $(a',b')$ is contained in each $S_i$: $(a',b') \in S_1 \cap S_2$ follows from the choice $(z_1,z_2)=(a,b)$. To see that $(a',b') \in S_3$, set $(z_1,z_2,z_3,z_4) =(a,b,b',c)$, where $c \in W$ is such that $R(c,b')$. \co{The analogous argument flipping the roles of $a$ and $b$   
shows that $(b',a') \in S_3$}, and \co{hence}  $(a',b') \in S_4$. Similarly, $(b',a')$ is also contained in each $S_i$.

Now we claim that $S_1$ does not intersect $U \times V$. Indeed, if $x \in U$, $y\in G$, and $S_1(x,y)$ is witnessed by $z_1,z_2\in G$ via the defining formula, then $z_1 \in V$, since $x \in U$, $z_1 \in U\cup V$, $R(x,z_1)$, and $\relstr{G}$ has no pseudoloops. Thus, since $\sim$ satisfies (i) and $z_1 \sim z_2$, we must also have $z_2 \in V$. Hence, we cannot have $y\in V$, since $R(z_2,y)$ and since $\relstr{G}$ has no pseudoloops. The same reasoning shows that $S_1$ does not intersect $V \times U$, that $S_2$ intersects neither $U \times W$ nor $W \times U$, that $S_3$ does not intersect $V \times W$, and that $S_4$ does not intersect $W \times V$. It follows that $S = S_1 \cap S_2 \cap S_3 \cap S_4$ satisfies (i).

In summary, $S$ is reflexive, satisfies (i), and contains both $(a',b')$ and $(b',a')$.
We now define the relation $Q$ by
$$
Q(x,y) :\leftrightarrow \exists s \; (x \sim s \wedge S(s,y)) \wedge \exists t \; (S(t,x) \wedge t \sim y).
$$
By definition, $Q$ is symmetric, and it satisfies (i) as both $\sim$ and $S$ do.
From the reflexivity of $S$, it follows that $Q$ contains $\sim$; in particular, it is itself reflexive. Moreover, the reflexivity of $\sim$ together with $(a',b'), (b',a') \in S$ imply that  $(a',b')\in Q$.

Finally, by repeatedly replacing $Q$ by $Q'(x,y) :\leftrightarrow \exists z \; Q(x,z) \wedge Q(z,y)$ we obtain an increasing chain of reflexive symmetric relations. Using again the oligomorphicity of $\group{G}$, this chain stabilizes after finitely many steps and we get the transitive closure of the original $Q$ -- the sought after equivalence relation  $\sim'$. 

\bigskip

\noindent \textit{Step 2:}
Let $(\relstr{G}^q,\group{G}^q)$, where $\relstr{G}^q = (G^q; R^q)$, be the quotient gg-system modulo $\sim$: that is, $G^q$ consists of the equivalence classes of $\sim$, $\group{G}^q$ is the natural action of $\group{G}$ on these classes, and $R^q(U,V)$ holds for classes $U,V$ if $R(u,v)$ holds for some $u\in U$ and $v\in V$.

Note that $(\relstr{G},\group{G})$ pp-interprets $\relstr{G}^q$ and that $\relstr{G}^q$ is a connected smooth digraph of algebraic length 1. Moreover, $\relstr{G}^q$ has no pseudoloops since $\sim$ satisfies (i). From (ii) (or (iii)) it follows that 
if $R^q(U,V)$ and $R^q(U,W)$ (or $R^q(V,U)$ and $R^q(W,U)$) where $V$ is in the same $\group{G}^q$-orbit as $W$, then $V = W$.

We show that $G^q$ is finite. Fix any vertex $U \in G^q$ and consider any $V,W \in G^q$ such that $(U,V)$ is in the same $\group{G}^q$-orbit of pairs as $(U,W)$, say $\alpha(U,V) = (U,W)$, where $\alpha \in \group{G}^q$. Recall that $\relstr{G}^q$ is connected and consider any oriented path $U=U_0, U_1, \dots, U_k=V$ and its $\alpha$--image $U=\alpha(U_0), \alpha(U_1), \dots \alpha(U_k) = W$. Clearly, each $U_i$ is in the same orbit as $\alpha(U_i)$ and then, using the property of $R^q$ from the previous paragraph, we obtain $\alpha(U_1)=U_1$, \dots, $\alpha(U_k) = U_k$, so $V=W$. We have proved that $(U,V)$ and $(U,W)$ are in different orbits of pairs whenever $V \neq W$. Since $\group{G}^q$ is oligomorphic, we conclude that $G^q$ must be finite. 

\bigskip
\noindent \textit{Step 3:} 
Now $\relstr{G}^q$ is a finite smooth digraph of algebraic length 1 without a loop (in fact, it does not even have a pseudoloop).  By Theorem~\ref{thm:loop-lemma}, the quotient $\relstr{G}^q$ pp-interprets $\relstr{K}_3$ with parameters. Since $(\relstr{G},\group{G})$ pp-interprets $\relstr{G}^q$, the proof is complete.  
\end{proof}

Some assumptions in Proposition~\ref{prop:3orbits} can be weakened. For example, the connectivity assumption can be omitted and more quotient digraphs (other than $\relstr{K}_3$) can be allowed. We have decided to present the simplest version just to illustrate the general idea.

A theorem recently announced by  Marcin Kozik suggests that even the following strengthening of Conjecture~\ref{conj:loop} might be true. 

\begin{conjecture}\label{conj:loop-two}
Let $(\relstr{G},\group{G})$ be an oligomorphic gg-system such that the quotient of $\relstr{G}$ modulo the $\group{G}$-orbit equivalence is a smooth digraph of algebraic length $1$ without a loop. 
Then $(\relstr{G},\group{G})$ {pp}-interprets $\relstr{K}_3$ with parameters.
\end{conjecture}

Kozik proved this conjecture in the finite case. The infinite version may be of use for, e.g., answering the original question about (continuous) homomorphisms to projections discussed after Conjecture~\ref{conj:newdicho} --- the implication (2) $\Rightarrow$ (1) in Corollary~\ref{cor:homos}.

\subsection{Discussion of the pseudo-Siggers theorem}
\lib{Theorem~\ref{thm:main_fin} gives several equivalent conditions to the existence of a Siggers operation (and many more exist in the literature). Are there analogues in the $\omega$-categorical setting?}

A positive answer to Conjecture~\ref{conj:loop}, at least under the assumption that $R$ contains edges $(r,a), (a,r), (r,e), (e,a)$ for some $a,e,r 
\in G$,  would allow a strengthening of item~(iii) of Theorem~\ref{thm:main} to a $4$-variable pseudo-Siggers operation $s$ satisfying $\alpha s(r,a,r,e) \approx \beta s(a,r,e,a)$ for some unary functions $\alpha,\beta$ (which works in the finite case, see~\cite{Sig10,KMM14}).
Another open problem is whether it is possible to replace item (iii) of Theorem~\ref{thm:main} by a pseudo-weak-near-unanimity operation, i.e., an operation $w(x_1,\ldots,x_n)$, for some $n\geq 2$, satisfying $\alpha_1 w(x,\ldots,x,y) \approx \alpha_2 w(x,\ldots,x,y,x)\approx\ldots\approx \alpha_n w(y,x\ldots,x)$, for unary functions $\alpha_1,\ldots,\alpha_n$ (which again works in the finite, see~\cite{MM08}).

On the negative side, it has been observed that the CSP classification for the  reducts of $(\relstr{Q};<)$ in~\cite{tcsps-journal} shows that
the syntactically strongest characterization of (iii) in Theorem~\ref{thm:main} in the finite case by means of \emph{cyclic operations} (see~\cite{Cyclic}) cannot be lifted to the infinite, at least not in the straightforward way of adding unary functions. For a well-known  example, the clone of those functions in $\Pol(\relstr{Q};<)$ which are injective (up to dummy variables) contains a pseudo-Siggers operation~\cite{tcsps-journal} (and  the CSP of the corresponding structure is in P), but no operation $c(x_1,\ldots,x_n)$, where $n\geq 2$, satisfying $\alpha c(x_1,\ldots,x_n) \approx \beta c(x_2,\ldots,x_n,x_1)$: if there was such an operation, then picking any non-constant tuple $(a_1,\ldots,a_n)\in\relstr{Q}^n$, there would be an automorphism $\gamma$ of $(\relstr{Q};<)$ such that $c(b_1,\ldots,b_n)=\gamma c(b_2,\ldots,b_n,b_1)$ for all 
$b_1,\ldots, b_n\in \{a_1,\ldots,a_n\}$, since $(\relstr{Q};<)$ is a core structure.  But then shifting $n$ times we would get $c(a_1,\ldots,a_n)=\gamma^n c(a_1,\ldots,a_n)$, and so $\gamma$ would fix $c(a_1,\ldots,a_n)$, contradicting the injectivity of $c$.

Finally, we remark that our main theorem as well as the pseudoloop lemma can be used as tools for proving hardness. \lib{We are not aware of any such applications so far, but let us state these NP-hardness criteria for future reference.}


\lib{
\begin{corollary}
Let $\relstr{A}$ be an $\omega$-categorical core structure.
Assume that $\relstr{A}$ pp-interprets with parameters a pseudoloop-free graph containing a $\relstr{K}_3$, or equivalently, that $\relstr{A}$ does not have a pseudo-Siggers polymorphism. Then $\CSP(\relstr{A})$ is $\NP$-hard.
\end{corollary}
}

\section{Mappings to the projection clone}\label{sect:homos}

A recent revision of the basic reductions between CSPs  in~\cite{wonderland} has revealed the importance of a different kind of mapping -- \emph{uniformly continuous h1 clone homomorphisms} -- which is obtained from continuous clone homomorphisms by strengthening the topological requirement and weakening the algebraic one. We now briefly introduce these concepts and refer to~\cite{wonderland} for more details.

The pointwise convergence topology of a function clone is induced by a  \emph{uniformity}~\cite{uniformbirkhoff} which \co{in turn}  is, in the case of a countable domain, induced by a metric~\cite{Reconstruction}. An arity-preserving mapping $\xi$ from a clone $\clone{C}$ on a set $C$ to a  clone $\clone{D}$ on a finite set $D$ is uniformly continuous with respect to this \lib{uniformity} if there exists a finite subset $C' \subseteq C$ such that $\xi(f) = \xi(g)$ whenever $f$ and $g$ agree on $C'$.  
An \emph{h1 clone homomorphism} is a mapping from one function clone to another which preserves arities and composition with projections; equivalently, preserves identities of height $1$, i.e., identities of the form $f(\mbox{variables} ) \approx g(\mbox{variables})$. For example, the pseudo-Siggers identity of our main theorem is not of height~1, but becomes height~1 when the outer unary functions are removed.

These notions give us 12=$3 \times2 \times 2$ different meaningful types of morphisms from a clone to the clone of projections:  we may require that the mapping is uniformly continuous, or continuous, or we do no impose any topological condition at all; we may require that the mapping is a clone homomorphism or only an h1 clone homomorphism; finally, we can ask for a total mapping from the clone itself, or only for a mapping from some stabilizer. The trivial implications between the existence of such morphisms form a 3 $\times$ 2 $\times$ 2 grid. However, as it turns out, for polymorphism clones of $\omega$-categorical core structures there are at most six non-equivalent conditions and the implications among them form a chain.

\begin{corollary} \label{cor:homos}
Consider the following statements for an $\omega$-categorical core $\relstr{A}$.
\begin{enumerate}
\item[(1)] $\Pol(\relstr{A})$ has a uniformly continuous clone homomorphism to $\trivclone$.
\item[(1')] $\Pol(\relstr{A})$ has a continuous clone homomorphism to $\trivclone$.
\item[(2)] $\Pol(\relstr{A})$ has a clone homomorphism to $\trivclone$.
\item[(3)] Some $\Pol(\relstr{A},c_1,\ldots,c_n)$ has a clone homomorphism to $\trivclone$.
\item[(3')] Some $\Pol(\relstr{A},c_1,\ldots,c_n)$ has a continuous clone homomorphism to $\trivclone$.
\item[(3'')] Some $\Pol(\relstr{A},c_1,\ldots,c_n)$ has a uniformly continuous clone homomorphism to $\trivclone$.
\item[(4)] Some $\Pol(\relstr{A},c_1,\ldots,c_n)$ has a uniformly continuous  h1 clone homomorphism to $\trivclone$. 
\item[(4')] $\Pol(\relstr{A})$ has a uniformly continuous  h1 clone homomorphism to $\trivclone$. 
\item[(5)] $\Pol(\relstr{A})$ has a continuous  h1 clone homomorphism to $\trivclone$. 
\item[(5')] Some $\Pol(\relstr{A},c_1,\ldots,c_n)$ has a continuous  h1 clone homomorphism to $\trivclone$. 
\item[(6)] $\Pol(\relstr{A})$ has an  h1 clone homomorphism to $\trivclone$. 
\item[(6')] Some $\Pol(\relstr{A},c_1,\ldots,c_n)$ has an  h1 clone homomorphism to $\trivclone$. 
\end{enumerate}
Then all statements with equal number are equivalent, and (i) implies (j) for all $1\leq i\leq j\leq 6$.
\end{corollary}
\begin{proof}
Items (1) and (1') are equivalent by the proof of Lemma 20 in~\cite{Topo-Birk} (cf.~\cite{uniformbirkhoff} for an explicit proof thereof), (3) and (3') are equivalent by Theorem~\ref{thm:main}, and (3') and (3'') again by the proof in~\cite{Topo-Birk}. Items (4) and (4'), (5) and (5'), as well as (6) and (6') are equivalent by Corollary 8.1 in~\cite{wonderland} which implies that there is a uniformly continuous h1 clone homomorphism from $\Pol(\relstr{A})$ to $\Pol(\relstr{A},c_1,\ldots,c_n)$ (although the argument presented there does not work; a correct argument is obtained by combining Lemma 3.9 \co{of that article with Corollary 4.7 thereof)}. Given these equivalences
it is clear that the strength of the statements is decreasing with increasing number.
\end{proof}

In~\cite{wonderland} it was shown that (4') still implies $\NP$-hardness of the CSP. This challenged the original dichotomy conjecture of Bodirsky and Pinsker, Conjecture~\ref{conj:dicho}, which identified NP-hardness with (3'). More on the eventual resolution of this matter in the next subsection.

Trung Van Pham has observed that our main theorem implies that the failure of (4') is witnessed by the satisfaction of the pseudo-Siggers identity even for $\omega$-categorical structures which are not cores. In the following proof, we assume that the reader is familiar with the basic facts about the core of an $\omega$-categorical structure, and refer to~\cite{BKOPP, BKOPP-equations} for more details. 
\begin{corollary}
Let $\relstr{B}$ be $\omega$-categorical. If $\Pol(\relstr{B})$ has no uniformly continuous h1 clone homomorphism to $\trivclone$, then $\Pol(\relstr{B})$ contains a pseudo-Siggers operation.
\end{corollary}
\begin{proof}
Let $\relstr{A}$ be the core of $\relstr{B}$. Since  $\Pol(\relstr{B})$ has no uniformly continuous h1 clone homomorphism to $\trivclone$, neither does $\Pol(\relstr{A})$, by the results from~\cite{wonderland}. Hence, by Corollary~\ref{cor:homos} and Theorem~\ref{thm:main}, $\Pol(\relstr{A})$ has a pseudo-Siggers operation $s(x_1,\ldots,x_6)$ whose pseudo-Siggers identity is witnessed by unary functions $\alpha,\beta\in\Pol(\relstr{A})$.

To show that $\Pol(\relstr{B})$ contains a pseudo-Siggers operation as well, we show that it contains local pseudo-Siggers operations and then refer to Lemma~\ref{lem:localglobal}. So let $F\subseteq B$ be finite. 
We may assume that $A\subseteq B$ and that $\relstr{A}$ is an induced substructure of $\relstr{B}$. Let $h$ be a homomorphism from $\relstr{B}$ into $\relstr{A}$. Let $F'\subseteq A$ be the image of $F^6$ under the operation $s'(x_1,\ldots,x_6):=s(h(x_1),\ldots,h(x_6))\in\Pol(\relstr{B})$, and let $\gamma$ be an automorphism of $\relstr{A}$ which agrees with $h$ on $F'$; this automorphism exists since $\relstr{A}$ is a core and since $h$ restricted to $A$ is an endomorphism of $\relstr{A}$. Then for all \co{$x,y,z\in F$} we have 
$$
\alpha \gamma^{-1} h s'(x,y,x,z,y,z)= \beta \gamma^{-1} h s'(y,x,z,x,z,y)\; .
$$
Since $\alpha\gamma^{-1}h$ and $\beta\gamma^{-1}h$ are contained in $\Pol(\relstr{B})$, we see that the pseudo-Siggers identity is indeed satisfied by $s'$ on $F$.
\end{proof}


\subsection{Converse implications}

The implication (2) $\Rightarrow$ (1) is the original question from~\cite{BPP-projective-homomorphisms} about clone  homomorphisms to projections, which we already \co{briefly} discussed in the present paper after Conjecture~\ref{conj:newdicho}. It remains open.  

A counterexample to (3) $\Rightarrow$ (2) can be found  among the first-order reducts of $(\relstr{Q};<)$, investigated in~\cite{tcsps-journal}: it is the structure of \emph{shuffle closed temporal relations}, or equivalently, those relations preserved by a binary operation called  \emph{pp}, which essentially behaves like one projection on one half plane, and like the other projection on the remaining half plane. It is not hard to see that it satisfies (3) but not (2).

The following example, a simplified version of a construction communicated to us by Ross Willard, shows that the implication from (2) to (3) cannot be reversed even for finite structures. The example is a disjoint union of two copies of $\relstr{K}_3$ with an additional relation ensuring that the structure is a core.
 
\begin{example}
Consider $\relstr{A} = (A; R,S)$, where
\begin{align*}
A &:= \{1_0, 2_0, 3_0, 1_1, 2_1, 3_1\}, \\
R(a_i,b_j) &:\leftrightarrow  (i=j) \wedge (a \neq b), \mbox{ and } \\
S(a_i,b_j) &:\leftrightarrow  i \neq j.
\end{align*}
The structure $\relstr{A}$ is a core. Indeed, for each endomorphism $\alpha$ of $\relstr{A}$ and each pair of elements $a_i, b_j \in \relstr{A}$ with $\alpha(a_i)=\alpha(b_j)$ we have $i=j$ since $\alpha$ preserves $S$, and then $a=b$ since $\alpha$ preserves $R$. Therefore, each endomorphism of $\relstr{A}$ is injective, and hence an automorphism as $\relstr{A}$ is finite.

Next observe that, for an arbitrary chosen $c \in A$, the substructure of $(A; R)$ induced by the set $\{x\;|\; S(x,c)\}$, which is obviously pp-definable with a parameter, is isomorphic to $\relstr{K}_3$. It follows that $\relstr{A}$  pp-interprets $\relstr{K}_3$ with parameters and therefore a stabilizer of $\Pol(\relstr{A})$ has a clone homomorphism to $\trivclone$. 

Finally, in order to show that $\Pol(\relstr{A})$ does not have a homomorphism to $\trivclone$, we define polymorphisms $\alpha$ and $s$, where $\alpha$ is  unary and $s$ ternary, satisfying a \emph{non-trivial} identity: that is, an identity not satisfiable by members of $\trivclone$. 
Set
$$
\alpha(a_i) := a_{1-i} \quad \mbox{ and } \quad 
s(a_i,b_j,c_k) := \left\{ 
\begin{array}{ll}
c_k & \mbox{ if } i=j \\
a_i & \mbox{ if } i \neq j
\end{array}
\right. .
$$
Now $\alpha$, $s$ are clearly polymorphisms \co{of $\relstr{A}$} and they satisfy the nontrivial identity
$$
s(x,x,y) \approx s(y,\alpha(y),x)
$$
since both sides are equal to $y$ by the definitions.
\end{example}

A counterexample to (4) $\Rightarrow$ (3) is given in~\cite{BKOPP}. However, it is also proved in the same paper that each such counterexample must have at least double exponential growth (in $n$) of the number of orbits of $n$-tuples. This is never the case for reducts of finitely bounded homogeneous structures, so that (3) and (4) are equivalent in the range of Conjectures~\ref{conj:dicho} and~\ref{conj:newdicho}, and the conjectures remain plausible even \co{in light of} the above-mentioned results from~\cite{wonderland} which show that (4) still implies NP-hardness. 

No counterexamples are known for (5) $\Rightarrow$ (4) or (6) $\Rightarrow$ (5).
We suspect that (6) $\Rightarrow$ (4) at least for reducts of finitely bounded homogeneous structures. Note that if Conjecture~\ref{conj:dicho} is true and (6) does not imply (4), then we would be in a rather peculiar situation: The complexity of a CSP would depend only on the height 1 identities together with the uniform structure of the polymorphism clone (this fact follows from~\cite{wonderland}), but height 1 identities per se would be insufficient to decide the complexity. 
On the other hand, the topology would be irrelevant when all the identities (or at least the pseudo-identities as in our main theorem) are taken into account. 

We now provide an example showing that the implication (6) $\Rightarrow$ (5) does not hold in general for polymorphism clones. Our example answers Question~7.2 in~\cite{BPP-projective-homomorphisms} to the negative, by the results in~\cite{wonderland,uniformbirkhoff,Topo-Birk}.

\begin{example}
Define for every $n\geq 1$ a binary function $g_n\colon \omega^2\To\omega$ as follows: 
$$
g_n(x,y):= \left\{ 
\begin{array}{ll}
n & \mbox{ if } x<n \\
x & \mbox{ otherwise.}
\end{array}
\right.
$$
Let $\clone{C}$ be the smallest function clone of the domain $\omega$ which contains the set  $\{g_n\;|\; n\geq 1\}$; then $\clone{C}$ consists of all term functions over this set.

Let $f\in \clone{C}$, and denote its arity by $m$. We claim that there exists $k_f\geq 0$ such that $f(x_1,\ldots,x_m)\geq k_f$ for all $x_1,\ldots,x_m \in [0,k_f)$, and such that $f$ restricted to $[k_f,\infty)^m$ equals a projection. To see this, we use induction over terms. Clearly,  the claim is true for the functions $g_n$ (set $k_{g_n}:=n$) and all projections. Now let $$f(x_1,\ldots,x_m)=g_n(s(x_1,\ldots,x_m),t(x_1,\ldots,x_m)),$$ where $n\geq 1$, and $s,t\in \clone{C}$ satisfy the claim.
\begin{itemize}
\item If $x_1,\ldots,x_m \in [0,k_s)$, then $f(x_1,\ldots,x_m)\geq s(x_1,\ldots,x_m)\geq k_s$ since $g_n(x,y)\geq x$ for all $x,y\in\omega$ by definition. 

\item If $x_1,\ldots,x_m \in [0,n)$, then either $s(x_1,\ldots,x_m)\geq n$ and hence 
$$f(x_1,\ldots,x_m)\geq s(x_1,\ldots,x_m)\geq n,$$ or $s(x_1,\ldots,x_m)<n$, in which case we have  $g_n(s(x_1,\ldots,x_m),t(x_1,\ldots,x_m))= n$ by the definition of $g_n$.

\item On $[\max\{k_s,n\},\infty)^m$ we have that $s(x_1,\ldots,x_m)$ behaves like a projection, by the definition of $k_s$. Therefore, for $x_1,\ldots,x_m\in [\max\{k_s,n\},\infty)$ we have $s(x_1,\ldots,x_m)\in\{x_1,\ldots,x_m\}$, and in particular $s(x_1,\ldots,x_m)\geq n$. Thus,  $$g_n(s(x_1,\ldots,x_m),t(x_1,\ldots,x_m))=s(x_1,\ldots,x_m).$$ Consequently, $f$ behaves like the same projection as $s$ on $[\max\{k_s,n\},\infty)^m$.
\end{itemize}
Summarizing, we obtain that setting $k_f:=\max\{k_s,n\}$ proves the claim.


The mapping $\xi$ which sends every function $f(x_1,\ldots,x_m)\in \clone{C}$ to the projection which it equals on $[k_f,\infty)^m$ clearly is a clone homomorphism from  $\clone{C}$ to $\trivclone$. 

However, the clone $\clone{C}$ does not enjoy any uniformly continuous h1 clone homomorphism $\xi'$ to $\trivclone$: uniform continuity would imply that there exists a finite set $F\subseteq \omega$ such that the value of every function in $\clone{C}$ under $\xi'$ only depends on its restriction to $F$. But then $\xi'(g_n(x,y))=\xi'(g_n(y,x))$ for any $n\geq 1$ such that $F\subseteq [0,n)$, a contradiction since $\xi'(g_n(y,x))$ is, by the definition of an h1 clone homomorphism, also the opposite projection of $\xi'(g_n(x,y))$ (as such homomorphisms preserve the switching of variables).

The topological closure $\overline{\clone{C}}$ of $\clone{C}$ in the space of all finitary functions on $\omega$ is a polymorphism clone. \co{Containing $\clone{C}$, the clone $\overline{\clone{C}}$ has no uniformly continuous h1 clone homomorphism to $\trivclone$, and thus does not satisfy~(5). Moreover,} every $f\in \overline{\clone{C}}$ still enjoys the property of the claim above: for any sequence $(f_n)_{n\in\omega}$ in $\clone{C}$ converging to $f$ we must have that $(k_{f_n})_{n\in\omega}$ is eventually constant, and so setting $k_f$ to that value works. Hence, $\overline{\clone{C}}$ has a clone homomorphism to $\trivclone$ as well, \co{and in particular an h1 clone homomorphism, showing that it satisfies~(6).}

We remark that the mapping $\xi$ is continuous: if a sequence $(f_n)_{n\in\omega}$ in $\overline{\clone{C}}$ converges to $f\in \overline{\clone{C}}$, then we must have $k_f=k_{f_n}$ eventually, and that all $f_n$ behave like the same projection on $[k_f,\infty)^m$ eventually.\smallskip

\end{example}


\subsection{Taylor equations without idempotency}

Long before the CSP motivated the intensive investigation of function clones that do not have a clone homomorphism to $\trivclone$, such clones were characterized by Walter Taylor by means of certain identities~\cite{T77}. His theorem does not require any finiteness or topological assumptions. On the other hand, it only concerns \emph{idempotent} clones, that is, clones whose all operations $f$ satisfy the identity $f(x,x\co{,} \dots,x) \approx x$ (or, in other words, all operations are polymorphisms of all singleton unary relations).

\begin{theorem}[\cite{T77}] \label{thm:Taylor}
Let $\clone{C}$ be an idempotent function clone. Then $\clone{C}$ does not have a clone homomorphism to $\trivclone$ if and only if it contains an operation $t(x_1,\ldots,x_n)$ for some $n\geq 1$ such that, for any $1\leq i\leq n$, $t$ satisfies an identity of the form $t(\dots,x,\dots) \approx t(\dots,y,\dots)$, where $x$ and $y$ are at the $i$-th position (and the remaining positions contain some variables which are not further specified).
\end{theorem}

The standard proof (see, e.g., \cite{HM88}) of the non-trivial  implication (i.e., from left to right) can be split into two parts. First, every h1 clone homomorphism from an idempotent clone to $\trivclone$ is, in fact, a clone homomorphism. Second, the absence of $t$ as in Theorem~\ref{thm:Taylor} implies the existence of an h1 clone homomorphism by a compactness argument. While the former part can be generalized to the non-idempotent situation as proved in the following proposition, Example~\ref{ex:non-taylor} below shows that the latter part does not generalize in a straightforward fashion.  

For the purpose of the next proposition, we say that a mapping  $\xi\colon \clone{C} \to \clone{D}$, where $\clone{C}, \clone{D}$ are function clones, is an \emph{almost clone homomorphism} if it is an h1 clone homomorphism and it preserves the composition with unary members both from inside and outside; equivalently, $\xi$ preserves arities and all identities of the form
$$
\alpha (f ( \beta_1(x_{i_1}), \dots, \beta_n(x_{i_n}))) \approx
\gamma (g ( \delta_1(x_{j_1}), \dots, \delta_n(x_{j_n}))),
$$
where some of the unary functions of the identity can be missing (so that in particular, h1 identities are of this form).

\begin{proposition}\label{prop:almost}
Let $\clone{C}$ be a function clone and let $\xi\colon \clone{C} \to \trivclone$ be a mapping.
Then $\xi$ is a clone homomorphism if and only if it is an almost clone homomorphism.
\end{proposition}

\begin{proof}
Assume that $\xi$ is an almost clone homomorphism. We will show that $\xi$ is a clone homomorphism in a sequence of steps. The first step is well-known.

\bigskip
\noindent \textit{Step 1:} It is enough to show that $\xi$ preserves the projections and identities of the form
$$
f(x_1, \dots, x_{k}, g(x_{k+1}, \dots, x_{l}), x_{l+1}, \dots, x_m) \approx h(x_1, x_2, \dots, x_m).
$$
This follows from the fact that each term can be obtained by repeatedly using the type of composition on the left hand side and then merging the variables.

\bigskip

\noindent \textit{Step 2:} \co{We claim that} whenever an operation $f(x_1,\ldots,x_n) \in \clone{C}$  depends only on its $i$-th argument, then it is mapped via $\xi$ to the $i$-th $n$-ary projection; in particular, $\xi$ preserves projections. Indeed, such  $f$ satisfies the identity $$
f(x_1, \dots, x_n) \approx f(y_1, \dots, y_{i-1},x_i,y_{i+1}, \dots, y_n)
$$ and the only projection satisfying this identity is the projection onto the $i$-th coordinate. Since $\xi$ is an h1 clone homomorphism, the claim follows.

\bigskip

\noindent \textit{Step 3:} \co{We next prove that} the mapping $\xi$ preserves identities of the form
\begin{align*}
f(g(x_{1,1}, \dots, x_{1,m}), g(x_{2,1}, \dots, x_{2,m}), \dots, g(x_{n,1}, \dots, x_{n,m})) \approx \\
h(x_{1,1}, \dots, x_{1,m},x_{2,1}, \dots, \dots, x_{n,m})\; .
\end{align*}
Suppose $\xi(f)$ and $\xi(g)$ are the $i$-th and $j$-th projection, respectively. We need to show that
the above identity holds in $\trivclone$ after application of $\xi$, or in other words, that
$$
x_{i,j}\approx \xi(h)(x_{1,1}, \dots, x_{1,m},x_{2,1}, \dots, \dots, x_{n,m})
$$
holds in $\trivclone$.

Denoting $\alpha(x):=g(x,\dots, x)$, the following identity holds in $\clone{C}$:
$$
h(\underbrace{x_1, \dots, x_1}_{m \times },\co{\underbrace{x_2, \dots, x_2}_{m \times }}, \dots, \underbrace{x_{n}, \dots, x_{n}}_{m \times }) \approx f(\alpha(x_1), \alpha(x_2), \dots, \alpha(x_{n})) \; .
$$
Since $\xi$ preserves this identity and $\xi(\alpha)$ is the identity mapping, we obtain
$$
\xi(h)(\underbrace{x_1, \dots, x_1}_{m \times },\underbrace{x_2, \dots, x_2}_{m \times }, \dots, \underbrace{x_{n}, \dots, x_{n}}_{m \times }) \approx \xi(f)(x_1, x_2, \dots, x_{n}) \approx x_i
$$
and it follows that 
$$
x_{i,\ell}\approx \xi(h)(x_{1,1}, \dots, x_{1,m},x_{2,1}, \dots, \dots, x_{n,m})
$$
for some $1\leq \ell\leq m$.

Similarly, from
$$
h(x_1, \dots, x_{m},x_1, \dots, x_{m}, \dots) \approx \beta(g(x_1, \dots, x_{m})),
$$
where $\beta(x) := f(x, \dots, x)$, we get that 
$$
x_{k,j}\approx \xi(h)(x_{1,1}, \dots, x_{1,m},x_{2,1}, \dots, \dots, x_{n,m})
$$
for some $1\leq k\leq n$, proving \co{the claim of this step}.




\bigskip

\noindent \textit{Step 4:} Now we can finish the proof using Step~1. For simplicity, consider the case where $g$ appears at the first coordinate. So, suppose the following identity holds in $\clone{C}$. 
$$
f(g(x_{1}, \dots, x_{m}), x_{m+1}, \dots, x_n) \approx h(x_1, x_2, \dots, x_n).
$$
Set 
$$
t(x_{1,1}, \dots, x_{1,m}, x_{2,1}, \dots, \dots, x_{n,m}):=f(g(x_{1,1}, \dots, x_{1,m}),g(x_{2,1},\dots), \dots).
$$
From Step~3 we get $\xi(t)(x_{1,1}, \dots, x_{n,m}) \approx \xi(f)(\xi(g)(x_{1,1}, \dots), \dots)$, \co{and} in particular, 
\begin{align*}
\xi(t)(x_1, \dots, x_m,& \underbrace{x_{m+1}, \dots, x_{m+1}}_{m \times }, \ldots
,\underbrace{x_n, \dots, x_n}_{m 
\times }) \approx \\
&\approx \xi(f)(\xi(g)(x_1, \dots, x_m), \xi(g)(x_{m+1}, \dots, x_{m+1}), \dots, \xi(g)(x_n, \dots, x_n))\; .
\end{align*}
Since $\xi(g)$ is a projection, we have $\xi(g)(x, \dots, x) \approx x$.
Finally, from the definition of $t$ it follows that 
$$
t(x_1, \dots, x_m, x_{m+1}, \dots, x_{m+1}, \dots, x_n, \dots, x_n) \approx h(x_1, \dots, x_m, \alpha(x_{m+1}), \dots, \alpha(x_n)),
$$ 
where $\alpha(x) := g(x, \dots, x)$, so
$$
\xi(t)(x_1, \dots, x_m, x_{m+1}, \dots, x_{m+1}, \dots, x_n, \dots, x_n) \approx \xi(h)(x_1, \dots, x_m, x_{m+1}, \dots, x_n).
$$
Combining these facts we obtain 
\begin{align*}
\xi(f)&(\xi(g)(x_1, \dots, x_m), x_{m+1}, \dots, x_n) \\
&\approx \xi(f)(\xi(g)(x_1, \dots, x_m), \xi(g)(x_{m+1}, \dots, x_{m+1}), \dots, \xi(g)(x_n, \dots, x_n)) \\
&\approx \xi(t)(x_1, \dots, x_m, x_{m+1}, \dots, x_{m+1}, \dots, x_n, \dots, x_n) \\
&\approx  \xi(h) (x_1, \dots, x_n),
\end{align*}
which finishes the proof.

\end{proof}

The following example due to Miroslav Ol\v s\'ak shows that a direct analogue of Theorem~\ref{thm:Taylor}, or more precisely the above-mentioned second part of its standard proof, does not hold.

\begin{example} \label{ex:non-taylor}
Let $\clone{C}$ be the clone of all term operations of the free algebra over countably many generators in the signature consisting of unary operations $p,q$ and binary operations $r,s$ modulo the identities
$$
p r(x,y) \approx p s(x,y), \quad
q r(x,y) \approx q s(y,x).
$$
Since this system of identities is not satisfiable by projections, the clone $\clone{C}$ does not have a clone homomorphisms to $\trivclone$. 

Let $t\in \clone{C}$ be an arbitrary $n$-ary function. We claim that $t$  has a coordinate $1\leq k\leq n$ such that no identity of the following form is satisfied: 
$$
\alpha (t ( \beta_1(x_{i_1}), \dots, \beta_n(x_{i_n}))) \approx
\gamma (t ( \delta_1(x_{j_1}), \dots, \delta_n(x_{j_n}))),
$$
where ${i_k} \neq {j_k}$, and some of the unary functions can be missing in the identity.

Striving for a contradiction, suppose there was an $n$-ary $t\in \clone{C}$ satisfying such an identity for every $1\leq k\leq n$. Fix a term $t'(x_1,\ldots,x_n)$ over the variables $x_1, \ldots, x_n$ in the signature $\{p,q,r,s\}$ which evaluates in $\clone{C}$ as $t$.


Assume first that the outermost symbol of $t'$ is not $r$. We inductively define a mapping $\theta$ from the set of all terms over the variables $x_1, \ldots, x_n$ to the set $\{1, \dots, n\}$ as follows:
$$
\theta(z) = \left\{
\begin{array}{ll}
i & \mbox{ if } z = x_i \\
\theta(z_2) & \mbox{ if } z = q r( z_1, z_2) \\
\theta(z_1) & \mbox{ otherwise, and  $z = f(z_1, \dots, z_l)$ for some $f \in \{p,q,r,s\}$\; .}
\end{array}
\right.
$$
In other words, $\theta(z)$ is the index of its leftmost variable after rewriting it using the rule which replaces each subterm of the form $q r(z_1,z_2)$ by $q s(z_2,z_1)$. It is clear that $\theta$ preserves  composition with unary \co{operation} symbols from the inside and, for terms $z$  whose outermost symbol is not $r$, also from the outside. 

It is apparent from the defining identities of $\clone{C}$ that $\theta(z_1)=\theta(z_2)$ whenever $z_1$ and $z_2$ evaluate to the same function in $\clone{C}$. On the other hand, setting $k := \theta(t')$ we obtain 
\begin{align*}
\theta(\alpha (t' ( \beta_1(x_{i_1}), \dots, \beta_n(x_{i_n}))))
&=
\theta(t' (x_{i_1}, \dots, x_{i_n})) = i_k, \text{and}\\
\theta(\gamma (t' ( \delta_1(x_{j_1}), \dots, \delta_n(x_{j_n}))) &= j_k,
\end{align*}
a contradiction.

If the outermost symbol of $t'$ is $r$, then we adjust the definition of $\theta$ by swapping the roles of $r$ and $s$, and argue until a contradiction in the same fashion.
\end{example}

Note that any pseudo-Siggers operation satisfies the modification of Taylor identities (from Theorem~\ref{thm:Taylor}) which is described in the second paragraph of  Example~\ref{ex:non-taylor} (even without the inner unary functions). Hence, the non-satisfaction of identities of that form does, in the case of polymorphism clones of $\omega$-categorical core structures, imply the existence of a uniformly continuous clone homomorphism from  some stabilizer to $\trivclone$, and then also the existence of an h1 clone homomorphism to $\trivclone$ by Corollary~\ref{cor:homos}. 
The non-satisfaction of Taylor identities (without inner and outer unary functions) is not sufficient to this end: the function clone of all (up to dummy variables) injective functions on a countable set has no Taylor term (whose identities are incompatible with injectivity), but satisfies non-trivial h1 identities, e.g., those implicitly described in the statement of Lemma IV.3 in~\cite{BKOPP}.

\subsection{Ternary identities}\label{subsect:ternary}

Perhaps surprisingly in the light of our main theorem, it has been observed by Jakub Opr\v{s}al and independently by Ross Willard (and perhaps might have been known to others as well)  that the existence of clone homomorphisms to $\trivclone$ only depends on the structure of the ternary operations in the clone.

\begin{proposition}
Let $\clone{C}$ be a function clone, and let $\xi$ be a mapping from the ternary operations in $\clone{C}$ to $\trivclone$ which preserves arities and identities (h1 identities). Then $\xi$ has a unique extension to a clone homomorphism (h1 clone homomorphism) from $\clone{C}$ to $\trivclone$.
\end{proposition}
\begin{proof}
We give a self-contained proof, although the claim essentially follows from the results and arguments in~\cite{wonderland}, in particular, Section 7 thereof.

Let $f(x_1,\ldots,x_n) \in \clone{C}$, where $n\geq 1$, and let $a=(a_1,\ldots,a_n)\in\{0,1\}^n$. We define $f_{a}(x,y)\in\clone{C}$ to be the binary function obtained by identifying all variables $x_i$ of $f$ with $x$ whenever $a_i=0$, and all other variables with $y$. We further define an $n$-ary operation  $\xi'(f)$ on $\{0,1\}$ by setting $\xi'(f)(a):=\xi(f_a)(0,1)$ for all $a\in \{0,1\}^n$.

It is easy to verify that $\xi'$ is an extension of $\xi$, since $\xi$ preserves h1 identities on the ternary functions. Moreover, it is clear for the same reason that any extension of $\xi$ to an h1 clone homomorphism into $\trivclone$ must be compatible with the above definition, and hence be equal to $\xi'$, proving uniqueness of the extension.

We claim that $\xi'$ maps into $\trivclone$. To show this, we use the fact that $\trivclone=\Pol(\relstr{M})$, where
\[
\mathbb{M} = (\{0,1\}; \{(0,0,1),(0,1,0),(1,0,0)\}),
\]
as in the introduction. If $\xi'(f)$ was not a projection for some $f(x_1,\ldots,x_n)\in\clone{C}$, then there would exist triples $r_1,\ldots, r_n$ in the relation of $\mathbb{M}$ such that $\xi'(f)(r_1,\ldots,r_n)$, calculated componentwise, is not contained in that relation. Since the relation contains only three triples, we can define a ternary operation $g(x,y,z)$ from $f$ by replacing $x_i$ by $x$ whenever $r_i$ equals $(0,0,1)$, similarly for $y$ and $(0,1,0)$, and for $z$ and $(1,0,0)$. Then $\xi'(g)$ still violates $\mathbb{M}$, a contradiction since $\xi(g)=\xi'(g)$ and since $\xi(g)$ is a projection.

To see that $\xi'$ is an h1 clone homomorphism, consider an identity
$$
f(y_{i_1},\ldots,y_{i_n})\approx s(y_1,\ldots,y_m),
$$
and say that $\xi'(f)$ projects onto the $j$-th coordinate, where $1\leq j\leq n$. Then 
$$\xi'(f)(y_{i_1},\ldots,y_{i_n})\approx y_{i_j}$$ holds in $\trivclone$; we wish to show that $\xi'(s)(y_1,\ldots,y_m)\approx y_{i_j}$. Let $a=(a_1,\ldots,a_m)\in\{0,1\}^m$ be the tuple in which $a_{i_j}=1$, and all other $a_i=0$; since $\xi'(s)\in\trivclone$, it suffices to show that $\xi'(s)(a)=1$. By the definition of $\xi'$ and the first  identity above,  
we have that $\xi'(s)(a)=\xi(g)(0,1)$, where $g(x,y)$ is the binary function obtained from $f(x_1,\ldots,x_n)$ by inserting $y$ for $x_k$ whenever $i_k=i_j$, and $x$ otherwise. Again by the definition of $\xi'$ and the first identity, this value  is equal to $\xi'(f)(b)$, where $b=(b_1,\ldots,b_n)\in\{0,1\}^n$ is the tuple where $b_k=1$ if and only if $i_k=i_j$. By the second identity above, that value equals $1$, as desired. 

When $\xi$ preserves identities, then it is not hard to show that $\xi'$ is a clone homomorphism: one argues that it is an almost clone homomorphism by reducing to ternary functions, and then refers to Proposition~\ref{prop:almost}.
\end{proof}

A \co{standard} application of the compactness theorem of logic yields  the following corollary.

\begin{corollary}\label{prop:ternary2}
Let $\clone{C}$ be a function clone. Then there is no clone homomorphism (h1 clone homomorphism) $\clone{C}\To \trivclone$ if and only if $\clone{C}$ satisfies a non-trivial finite system of identities (h1 identities) of ternary functions.
\end{corollary}
In particular, our pseudo-Siggers identity, which involves a 6-ary function, always implies a non-trivial system of ternary identities -- but note that there is no bound on the number of identities in the system.

\section{Conclusion}\label{sect:conclusion}

The presented results show the possibility of a purely algebraic theory for function clones that are concerned by the infinite-domain tractability conjecture. They have already found several 
applications: in the classification of CSPs over reducts of unary structures~\cite{unaries},
and indirectly (via results in \cite{BKOPP} inspired by the present work) in the algebraic dichotomy for CSPs in MMSNP~\cite{MMSNP}.  Nevertheless, the theory is presently only in its infancy compared to its finite counterpart. We now briefly discuss three of the research directions that call for further exploration. 

Firstly, the technical core of our main result is the pseudoloop lemma which generalizes a loop lemma for finite undirected graphs. The extension of the latter result to directed graphs stated in Theorem~\ref{thm:loop-lemma} has \co{led} to the development of absorption theory~\cite{Cyclic,AbsorptionSurvey},  which in turn proved to be a powerful tool in the context of  finite domain CSPs and universal algebra. 
While it might be possible to lift some results from the finite to the infinite, e.g., in a manner akin to the one shown in Proposition~\ref{prop:3orbits}, it seems more valuable to find a genuine generalization of  absorption theory itself. Some ideas from this theory and, in particular, a loop lemma was used in a different infinite setting, namely that of idempotent clones~\cite{OlsakTerms}. Is it possible to merge these approaches?

Secondly, as discussed in Section~\ref{sect:homos}, there are several meaningful types of morphisms to the clone of projections and their exact relation in various settings (e.g., polymorphism clones of structures, polymorphism clones of $\omega$-categorical structures, polymorphism clones of reducts of homogeneous structures) requires further investigation. For finite core clones all the items (2) -- (6) in Corollary~\ref{cor:homos} are equivalent, and many alternative characterizations are available: for example by means of identities, the structure of invariant relations, absorption properties, the Tame Congruence Theory types, Bulatov's local algebraic properties, etc. Are there infinite analogues to these characterizations?

Thirdly, clone homomorphisms to other clones rather than only the projection clone $\trivclone$ are significant for CSPs as well as universal algebra. For example, the clones of modules play a similar role to $\trivclone$ in the context of solvability of finite-domain CSPs by consistency methods. Again, are there infinite analogues?

\bibliographystyle{plain}

\bibliography{irrelevant}
\end{document}